\documentclass{article}
\usepackage[margin=1in]{geometry}

\usepackage{amsthm}
\usepackage{amsmath}
\usepackage{amsfonts}
\usepackage{amssymb}
\usepackage{bm}
\usepackage{array}
\usepackage[caption=false,font=normalsize,labelfont=sf,textfont=sf]{subfig}
\usepackage{textcomp}
\usepackage{stfloats}
\usepackage{verbatim}
\usepackage{graphicx}

\usepackage[colorlinks=true, allcolors=red]{hyperref}

\usepackage{soul}
\usepackage{color}
\usepackage{xcolor}

\usepackage[noadjust]{cite} 

\usepackage{balance}

\allowdisplaybreaks

\newtheorem{lemma}{Lemma}
\newtheorem{remark}{Remark}
\newtheorem{openproblem}{Open Problem}
\newtheorem{corollary}{Corollary}

\newcommand{\ff}{{\mathbb F}}
\newcommand{\Fq}{\mathbb{F}_{q}}

\newcommand{\cA}{\mathcal{A}}

\newcommand{\cC}{\mathcal{C}}
\newcommand{\cF}{\mathcal{F}}
\newcommand{\cH}{\mathcal{H}}

\newcommand{\cK}{\mathcal{K}}
\newcommand{\cM}{\mathcal{M}}

\newcommand{\bI}{\mathbf{I}}

\newcommand{\bO}{\mathbf{O}}
\newcommand{\bP}{\mathbf{P}}

\newcommand{\bt}{\boxtimes}

\newcommand{\smat}{\ \ }
\newcommand{\rc}{\cdot}
\DeclareMathOperator{\diag}{diag}

\newlength{\spacelen}
\newcommand{\ws }[1]{\settowidth{\spacelen}{#1}\makebox[\spacelen]{}}
\newcommand{\wid}[3]{\settowidth{\spacelen}{#1}\makebox[\spacelen][#2]{$#3$}}
\newcommand{\wn}{\ws{$-$}}
\newcommand{\rL}[2][\ws{$-$}]{\wid{$-L_{00}$}{l}{#1 L_{#2}}}
\newcommand{\bzero}{\mathbf{0}}
\newcommand{\gbzero}{\color{gray}\mathbf{0}\color{black}}

\begin{document}
	
	\title{MSR Codes with Linear Field Size and Smallest Sub-packetization for Any Number of Helper Nodes}
	
	\author{
		Guodong~Li,
		Ningning~Wang,
		Sihuang~Hu, 
		and~Min~Ye
		\thanks{
			Research partially funded by
			National Key R\&D Program of China under Grant No. 2021YFA1001000,
			National Natural Science Foundation of China under Grant No. 12001322 and 12231014
			and a Taishan scholar program of Shandong Province.
		}
		\thanks{
			Guodong Li, Ningning Wang and Sihuang Hu are with Key Laboratory of Cryptologic Technology and Information Security, Ministry of Education, Shandong University, Qingdao, Shandong, 266237, China and School of Cyber Science and Technology, Shandong University, Qingdao, Shandong, 266237, China.
			S. Hu is also with Quan Cheng Laboratory, Jinan 250103, China.
			Email: \{guodongli, nningwang\}@mail.sdu.edu.cn, husihuang@sdu.edu.cn, yeemmi@gmail.com
		}
	}
	\date{}
	
	\maketitle
	
	\begin{abstract}
		An $(n, k, \ell)$ array code has $k$ information coordinates and $r = n - k$ parity coordinates, where each coordinate is a vector in $\mathbb{F}_q^{\ell}$ for some finite field $\mathbb{F}_q$.
		An $(n, k, \ell)$ MDS array code has the additional property that any $k$ out of $n$ coordinates suffice to recover the whole codeword.
		Dimakis \emph{et al.} considered the problem of repairing the erasure of a single coordinate and proved a lower bound on the amount of data transmission that is needed for the repair.
		A minimum storage regenerating (MSR) code with repair degree $d$ is an MDS array code that achieves this lower bound for the repair of any single erased coordinate from any $d$ out of $n-1$ remaining coordinates.
		An MSR code has the optimal access property if the amount of accessed data is the same as the amount of transmitted data in the repair procedure.
		
		The sub-packetization $\ell$ and the field size $q$ are of paramount importance in MSR code constructions.
		For optimal-access MSR codes, Balaji \emph{et al.} proved that $\ell \geq s^{\left\lceil n/s \right\rceil}$, where $s = d-k+1$.
		Rawat \emph{et al.} showed that this lower bound is attainable for all admissible values of $d$ when the field size is exponential in $n$.
		After that, tremendous efforts have been devoted to reducing the field size.
		However, so far, reduction to a linear field size is only available for $d\in\{k+1,k+2,k+3\}$ and $d=n-1$.
		In this paper, we construct the first class of explicit optimal-access MSR codes with the smallest sub-packetization $\ell = s^{\left\lceil n/s \right\rceil}$ for all $d$ between $k+1$ and $n-1$,
		resolving an open problem in the survey (Ramkumar \emph{et al.}, Foundations and Trends in Communications and Information Theory: Vol. 19: No. 4).
		We further propose another class of explicit MSR code constructions (not optimal-access) with an even smaller sub-packetization $s^{\left\lceil n/(s+1)\right\rceil}$ for all admissible values of $d$, making significant progress on another open problem in the survey.
		Previously, MSR codes with $\ell=s^{\left\lceil n/(s+1)\right\rceil}$ and $q=O(n)$ were only known for $d=k+1$ and $d=n-1$.
		The key insight that enables a linear field size in our construction is to reduce $\binom{n}{r}$ global constraints of non-vanishing determinants to $O_s(n)$ local ones, which is achieved by carefully designing the parity check matrices.
	\end{abstract}

	\section{Introduction}\label{sect:intro}
	
	{M}{aximum} Distance Separable (MDS) codes are widely used in storage systems because they provide the maximum failure tolerance for a given amount of storage overhead.
	With the emergence of large-scale distributed storage systems, the notion of \emph{repair bandwidth} was introduced to measure the efficiency of recovering the erasure of a single codeword coordinate.
	Following the literature on distributed storage, a codeword coordinate is referred to as a (storage) node, and the erasure of a codeword coordinate is called a node failure.
	Repair bandwidth is the amount of data that needs to be downloaded from the remaining nodes to repair the failed node.
	For an $(n,k)$ MDS array code, Dimakis \emph{et al.}~\cite{Dimakis10} showed that if we access $d$ out of $n-1$ remaining nodes, then at least $1/(d-k+1)$ fraction of data needs to be downloaded from each of the $d$ helper nodes.
	The range of $d$ is between $k+1$ and $n-1$.
	A minimum storage regenerating (MSR) code with {\em repair degree} $d$ is an MDS code that achieves the minimum repair bandwidth for the repair of any single node failure from any $d$ helper nodes.
	We further say that an MSR code has the \emph{optimal access} property if the amount of accessed data in the repair procedure is equal to the minimum repair bandwidth.
	
	The study of MSR codes is divided into two major categories---scalar MSR codes and MSR array codes.
	Reed-Solomon (RS) code is the most prominent example of scalar MDS codes, whose repair problem was first studied by Guruswami and Wootters  \cite{Guruswami16STOC, Guruswami16}.
	In particular, Guruswami and Wootters viewed each codeword coordinate as a vector over a subfield of RS codes' symbol field and proposed repair schemes with smaller bandwidth than the naive repair scheme.
	The degree of field extension is called sub-packetization.
	Building on the framework in \cite{Guruswami16STOC, Guruswami16}, Tamo \emph{et al.}~\cite{Tamo17RS} proposed an RS code construction with optimal repair bandwidth, which is a scalar MSR code.
	In contrast, for an $(n,k,\ell)$ MSR array code, each codeword coordinate is a vector in $\mathbb{F}_q^{\ell}$ for some field $\mathbb{F}_q$, where the parameter $\ell$ is sub-packetization.
	Compared to scalar MSR codes, their array code counterparts allow more flexibility in the code construction, which results in a smaller sub-packetization.
	Indeed, the best-known sub-packetization of MSR array code construction is $\exp(O(n))$ while the sub-packetization of (high-rate) scalar MSR codes needs to be $\exp(\Theta(n\log(n))$, according to a lower bound in \cite{Tamo17RS}.
	For this reason, we only consider MSR array codes in this paper.
	
	\subsection{Previous results on MSR codes with small sub-packetization}
	Numerous constructions of MSR codes have been proposed in recent years.
	The first explicit construction of MSR codes is the product matrix construction \cite{Rashmi11}.
	Yet it is limited to the low code rate regime.
	After that, explicit constructions of high-rate MSR codes were given in \cite{Tamo13, Wang16, Ye16, Ye16a, Sasid16, Raviv17, Tamo17RS, Li18, Duursma21, Vajha21}.
	A complete list of constructions and parameters of MSR codes can be found in an excellent survey paper~\cite{CIT-115}.
	In practical storage systems, the number of parity nodes $r=n-k$ is usually a rather small number in order to achieve a small storage overhead.
	In light of this, we consider the regime of fixed $r$ and growing $n$ throughout this paper.
	
	All the constructions mentioned above have exponentially large sub-packetization, which is in fact inevitable for MSR codes.
	More precisely, for MSR codes with repair degree $d=n-1$, Goparaju \emph{et al.} \cite{Goparaju14} proved that $\ell\ge \exp(\Theta(\sqrt{n}))$.
	This lower bound was later improved to $\ell\ge \exp(\Theta(n))$ by Alrabiah and Guruswami \cite{Alrabiah19}.
	Moreover, a tight lower bound was known for optimal-access MSR codes: It was shown in \cite{Tamo14} that $\ell \ge r^{\frac{k-1}{r}}$ for such codes with repair degree $d=n-1$; Balaji \emph{et al.} \cite{Balaji22} generalized this bound to a restricted class of optimal-access MSR codes for all admissible values of $d$.
	The lower bound in \cite{Balaji22} is $\ell \ge s^{\lceil\frac{n}{s}\rceil}$, where $s=d-k+1$.
	
	For $d=n-1$, papers~\cite{Ye16a, Sasid16, Li18} constructed
	optimal-access MSR codes with field size $q=O(n)$ and sub-packetization  $\ell=r^{\lceil\frac{n}{r}\rceil}$, matching the lower bound mentioned above.
	For general values of $d$, a construction of optimal-access MSR codes with $\ell=s^{\lceil\frac{n}{s}\rceil}$ was presented in~\cite{Rawat16}.
	However, the construction in~\cite{Rawat16} is not explicit and requires an exponentially large field size.
	For $d<n-1$, optimal-access MSR codes with the smallest sub-packetization $\ell=s^{\lceil\frac{n}{s}\rceil}$ and linear field size are only known for $d\in\{k+1,k+2,k+3\}$, which come from a recent paper~\cite{Vajha21}.
	Obtaining such code constructions for general values of $d$ was listed as an open problem in the survey paper \cite{CIT-115}:
	\begin{openproblem}\cite[Open Problem 3]{CIT-115}\label{op-access}
		Provide explicit constructions of optimal-access MSR codes having least-possible sub-packetization level, for all possible $(n, k, d)$, with $d<n-1$.
	\end{openproblem}
	
	If we do not require the optimal-access property, then the best-known sub-packetization level is $\ell=s^{\left\lceil n/(s+1)\right\rceil }$ for MSR codes with repair degree $d$.
	MSR codes with $\ell=s^{\left\lceil n/(s+1)\right\rceil }$ and linear field size are only known for $d=n-1$ (see~\cite{Wang16}) and $d=k+1$ (see~\cite{LWHY22}).
	In~\cite{Liu22, Zhang23, Li23}, explicit constructions of MSR codes with $\ell=s^{\left\lceil\frac{n}{2}\right\rceil}$ for any $d$ were presented.
	Another open problem in \cite{CIT-115} concerns MSR codes without the optimal-access requirement:
	\begin{openproblem}\cite[Open Problem 4]{CIT-115}\label{op-general}
		Construct MSR codes with least-possible sub-packetization level for all $(n$, $k$, $d)$.
		(There is no optimal-access requirement here).
	\end{openproblem}
 
	To the best of our knowledge, all known explicit $(n,k)$ MSR code constructions with linear field size and $d<n-1$ helper nodes are listsd in Table~\ref{tab:cmp}.
 
	\begin{table}[ht]
		\centering
		\begin{tabular}{|l|l|c|}
			\hline
			\multicolumn{3}{|c|}{Optimal-access}                                                                       \\
			\hline
			\hline
			Codes                      & Sub-packetization $\ell$                                    & Restrictions    \\
			\hline
			\cite[Section VIII]{Ye16}  & $s^{n}$                                                     &                 \\
			\hline
			\cite{CB}                  & $s^{n}$                                                     &                 \\
			\hline
			\cite{Liu22}               & $s^{\left\lceil n/2 \right\rceil}$                          &                 \\
			\hline
			\cite{Vajha21}             & $s^{\left\lceil n/s \right\rceil}$      (smallest possible) & $s\in\{2,3,4\}$ \\
			\hline
			Section~\ref{sect:code1}   & $s^{\left\lceil n/s \right\rceil}$      (smallest possible) &                 \\
			\hline
			\multicolumn{3}{c}{}                                                                                       \\
			\hline
			\multicolumn{3}{|c|}{Not optimal-access}                                                                   \\
			\hline
			\hline
			Codes                      & Sub-packetization $\ell$                                    & Restrictions    \\
			\hline
			\cite[Section IV]{Ye16}    & $s^{n}$                                                     &                 \\
			\hline
			\cite{LWHY22}              & $s^{\left\lceil n/3 \right\rceil}$      (best-known)        & $d=k+1$         \\
			\hline
			\cite{Zhang23} \cite{Li23} & $s^{\left\lceil n/2 \right\rceil}$                          &                 \\
			\hline
			Section~\ref{sect:code2}   & $s^{\left\lceil n/(s+1) \right\rceil}$    (best-known)      &                 \\
			\hline
			\multicolumn{3}{c}{}
		\end{tabular}
		\caption{Explicit $(n,k)$ MSR code constructions with linear field size and $d<n-1$ helper nodes, where $s=d-k+1$.
			``smallest possible'' means that $\ell$ achieves the theoretical lower bound; ``best-known'' means that $\ell$ is the smallest among all the existing constructions.}
		\label{tab:cmp}
	\end{table}
	
	\subsection{Our contribution and key technique}
	In this paper, we propose two classes of explicit MSR code constructions with linear field size $q=O_s(n)$ for any repair degree $d$ between $k+1$ and $n-1$,
	where $O_s$ means that we treat $s$ as a constant\footnote{We can do so because $s$ is upper bounded by $r$, which is assumed to be a constant.}.
	The first class is optimal-access and has sub-packetization $s^{\left\lceil n/s \right\rceil}$, matching the lower bound.
	This settles Open Problem~\ref{op-access}.
	The second class is not optimal-access but has a smaller sub-packetization $s^{\left\lceil n/(s+1)\right\rceil }$,
	which is the best-known sub-packetization among all the existing MSR code constructions.
	This constitutes significant progress on Open Problem~\ref{op-general}.
	
	We give an overview of our proof techniques here.
	We define an $(n, k, \ell)$ array code $\cC$ using the parity check equations
	\begin{equation}\label{eq:parity_check}
		\cC=\{(C_0, \dots, C_{n-1}): H_0C_0 + \cdots + H_{n-1}C_{n-1} = \bzero\},
	\end{equation}
	where each node $C_i = (C_i(0), C_i(1), \dots, C_{i}(\ell-1))^T$ is a column vector in $\Fq^{\ell}$,
	each $H_i$ is an $r\ell \times \ell$ matrix over $\Fq$, and $\bzero$ on the right-hand side is the all-zero column vector of length $r\ell$.
	
	The MDS property of $\cC$ is equivalent to the following $\binom{n}{r}$ \emph{global constraints} of non-vanishing determinants:
	\begin{equation}\label{eq:detMDS}
		\begin{aligned}
			& \det([H_{i_1} \smat  H_{i_2} \smat \cdots \smat H_{i_r}])\neq0 \text{~for~all~} \\& \cF = \{i_1, i_2, \dots, i_r\} \subseteq\{0,\dots,n-1\}.
		\end{aligned}
	\end{equation}
	We divide $n$ parity check submatrices $H_0, H_1,\dots,H_{n-1}$ into $n/s$ groups of size $s$ and
	$n/(s+1)$ groups of size $s+1$ in the first and the second construction, respectively.
	By carefully designing the inner structures of each $H_i$, we are able to make the matrices $[H_{i_1} \smat  H_{i_2} \smat \cdots \smat H_{i_r}]$ in global
	constraints~\eqref{eq:detMDS} block upper triangular, and
	reduce $\binom{n}{r}$ global constraints to
	$O_s(n)$ \emph{local (in-group) constraints} of non-vanishing determinants.
	This allows us to prove that a linear field size suffices for the MDS property of our construction.
	
	On the other hand, thanks to the recursive structure of $H_i$, the repair scheme for any single node failure follows directly from the MDS property.
	To elaborate, suppose that some node $C_i$ is failed, then we can do the repair as follows:
	Firstly, we split the failed node $C_i\in\Fq^{\ell}$ into $s$ new nodes, denoted as $\bar C_i^{\left\langle z\right\rangle }\in\Fq^{\ell/s}, 0\le z \le s-1$.
	For each surviving node $C_j$ where $j\neq i$, we download a fraction of $1/s$ symbols from $C_j$ to construct a new node $\bar C_j\in\Fq^{\ell/s}$.
	Due to the recursive structure of $H_i$, the method we used to prove the MDS property of $\cC$ can be generalized to show that
	$$(\bar{C}_0,\dots, \bar{C}_{i-1},\quad \bar{C}_{i}^{\left\langle 0\right\rangle }, \dots, \bar{C}_{i}^{\left\langle s-1\right\rangle },\quad \bar{C}_{i+1},\dots, \bar{C}_{n-1})$$
	also defines an $(n+s-1,k+s-1,\ell/s)$ MDS array code.
	Therefore, $(\bar{C}_i^{\left\langle z\right\rangle }, 0\le z \le s-1)$, that is $C_i$, can be recovered from any $d = k+s-1$ nodes in the set $\{\bar C_j : 0\le j \le n-1, j\neq i\}$,
	and this repair scheme achieves the cut-set bound.
	
	When discussing our $(n,k,\ell = s^{\lceil n/s \rceil})$ MSR code, we always assume that $n$ is divisible by $s$.
	For the case of $s\nmid n$, the construction can be easily obtained by removing some nodes from a slightly longer code in the divisible case.
	Similarly, when we discuss our $(n,k,\ell = s^{\lceil n/(s+1) \rceil})$ MSR code, we always assume that $n$ is divisible by $s+1$.
	
	The rest of this paper is organized as follows:
	In Section~\ref{sect:preliminary}, we present notations and essential terminology to establish a framework of our code constructions.
	In Section~\ref{sect:code1}, we present the construction of $(n, k, \ell = s^{n/s})$ MSR codes with optimal access property, prove its MDS property and demonstrate that the above class of MDS array codes achieves the optimal repair bandwidth for single node failure.
	In Section~\ref{sect:code2}, we propose another class of MSR codes, which are not optimal-access but with smaller sub-packetization level $\ell = s^{n/(s+1)}$.
	\section{Preliminaries}\label{sect:preliminary}
	
	Given the code length $n$, the code dimension $k$, and the repair degree $d$,
	we set $r = n-k$, $s=d-k+1$.
	The nodes $C_0, C_1, \dots, C_{n-1}$ are divided into $\bar n$ groups.
	For our first MSR code construction in Section~\ref{sect:code1}, we set $\bar n = \lceil n/s \rceil$, while for the second MSR code construction in Section~\ref{sect:code2}, we set $\bar n = \lceil n/(s+1) \rceil$.
	Then the sub-packetization level of both constructions can be written as $\ell = s^{\bar n}$.
	
	Let $\Fq$ be a finite field of order $q$.
	For a positive integer $m$, we represent $[m]$ as the set $\{0,1,\dots, m-1\}$,
	and let $\bI_{m}$ be the $m\times m$ identity matrix over $\Fq$.
	For an integer $a$, we define
	\begin{align*}
		a+[m] & := \{a+i: i\in [m]\}.
	\end{align*}
	For simplicity, we denote the vector $(x_0, x_1, \dots, x_{m-1})$ over $\Fq$ as $x_{[m]}$.
	Additionally, for a subset $\cA\subseteq [m]$, we define a subvector of $x_{[m]}$ as
	$$
	x_{\cA} := (x_i: i\in \cA).
	$$
	
	For an element $x\in \Fq$, and a positive integer $t$, we define a column vector of length $t$ as
	\begin{equation*}
		L^{(t)}(x)  := \left[
		\begin{array}{c}
			1       \\
			x       \\
			x^2     \\
			\vdots  \\
			x^{t-1} \\
		\end{array}
		\right].
	\end{equation*}
	
	To simplify notations, we
	introduce the following matrix operator $\bt$.
	For a matrix $A$ and  an $m\times n$ block matrix $B$ written as
	\begin{equation*}
		B = \begin{bmatrix}
			B_{11} & \cdots & B_{1n} \\
			\vdots & \ddots & \vdots \\
			B_{m1} & \cdots & B_{mn}
		\end{bmatrix},
	\end{equation*}
	we define
	\begin{equation*}
		A\bt B := \begin{bmatrix}
			A\otimes B_{11} & \cdots & A\otimes B_{1n} \\
			\vdots          & \ddots & \vdots          \\
			A\otimes B_{m1} & \cdots & A\otimes B_{mn}
		\end{bmatrix},
	\end{equation*}
	where $\otimes$ is the Kronecker product. Note that the result $A\bt B$ is dependent on
	how the rows and columns of $B$ are partitioned, and we will explicitly specify the partition every time we use this notation.
	
	\begin{remark}\label{remark:block_tensor}
		For any two matrices $M$ and $N$ it is well-known that $M\otimes N$ is permutation equivalent to $N\otimes M$ (c.f. \cite[page 7, equation (6)]{tensorpermutation}).
		By definition
		\begin{equation*}
			B\otimes A = \begin{bmatrix}
				B_{11} \otimes A & \cdots & B_{1n} \otimes A  \\
				\vdots           & \ddots & \vdots            \\
				B_{m1} \otimes A & \cdots & B_{mn} \otimes  A
			\end{bmatrix}.
		\end{equation*}
		For each block entry $B_{ij}$, there exist permutation matrices $P_i$ and $Q_j$ such that
		$$
		A\otimes B_{ij} = P_i(B_{ij} \otimes A)Q_j.
		$$
		Thus, $A\bt B$ is permutation equivalent to $B\otimes A$.
		Similarly, we can check that $A\bt B$ is also permutation equivalent to $A\otimes B$.
	\end{remark}
	
	Next, we define the following $s+1$ \emph{kernel maps}, which will play a central role in the design of parity check submatrices $H_0, H_1, \dots, H_{n-1}$ in \eqref{eq:parity_check}.
	For $b\in [s+1]$ and positive integer $t$,
	we define
	\begin{equation*}
		\cK_{b}^{(t)} : \Fq^{s} \to \Fq^{st\times s}
	\end{equation*}
	which maps a vector $x_{[s]}$ to the following $st\times s$ matrix $\cK_{b}^{(t)}(x_{[s]})$ over $\Fq$:
	We regard $\cK_{b}^{(t)}(x_{[s]})$ as an $s\times s$ block matrix, whose block entries are column vectors of length $t$,
	for $b\in[s+1]$ and $i,j\in [s]$, the block entry appearing at the $i$th block row and the $j$th (block) column is
	\begin{equation}\label{eq:cK}
		\cK_b^{(t)}(x_{[s]})(i,j) =
		\begin{cases}
			\wn L^{(t)}(x_j) & \text{ if } i=j,   \\[.3em]
			-L^{(t)}(x_j)    & \text{ if } i=b,   \\[.3em]
			\wn \bzero       & \text{ otherwise}.
		\end{cases}
	\end{equation}
	As the block row index $i$ is always less than $s$, for the case of $b=s$, indeed we have
	\begin{equation*}
		\cK_s^{(t)}(x_{[s]})(i,j) =
		\begin{cases}
			\wn L^{(t)}(x_j) & \text{ if } i=j,   \\[.3em]
			\wn \bzero       & \text{ otherwise}.
		\end{cases}
	\end{equation*}
	In the case of $s=3$, we have

	\begin{align*}
		\cK_0^{(t)}(x_{[3]})= &
		\begin{bmatrix}
			\wn L_{0}^{(t)} & -L_{1}^{(t)}    & -L_{2}^{(t)}    \\[.3em]
			\wn \gbzero     & \wn L_{1}^{(t)} & \wn \gbzero     \\[.3em]
			\wn \gbzero     & \wn \gbzero     & \wn L_{2}^{(t)}
		\end{bmatrix}, \\
		\cK_1^{(t)}(x_{[3]})= &
		\begin{bmatrix}
			\wn L_{0}^{(t)} & \wn \gbzero     & \wn \gbzero     \\[.3em]
			-L_{0}^{(t)}    & \wn L_{1}^{(t)} & -L_{2}^{(t)}    \\[.3em]
			\wn \gbzero     & \wn \gbzero     & \wn L_{2}^{(t)}
		\end{bmatrix}, \\
		\cK_2^{(t)}(x_{[3]})= &
		\begin{bmatrix}
			\wn L_{0}^{(t)} & \wn \gbzero     & \wn \gbzero     \\[.3em]
			\wn \gbzero     & \wn L_{1}^{(t)} & \wn \gbzero     \\[.3em]
			-L_{0}^{(t)}    & -L_{1}^{(t)}    & \wn L_{2}^{(t)}
		\end{bmatrix}, \\
		\cK_3^{(t)}(x_{[3]})= &
		\begin{bmatrix}
			\wn L_{0}^{(t)} & \wn \gbzero     & \wn \gbzero     \\[.3em]
			\wn \gbzero     & \wn L_{1}^{(t)} & \wn \gbzero     \\[.3em]
			\wn \gbzero     & \wn \gbzero     & \wn L_{2}^{(t)}
		\end{bmatrix},
	\end{align*}~where $L_j^{(t)} = L^{(t)}(x_j)$.
	
	\begin{remark} \textbf{1)} The block entry of $\cK_{b}^{(t)}(x_{[s]})$ is nonzero if and only if it lies on the main block diagonal or the $b$th block row,
		and the last block matrix  $\cK_{s}^{(t)}(x_{[s]})$ is block diagonal.
		\textbf{2)} In our construction of optimal-access MSR codes we only use the first $s$ kernel maps, while in the second construction of (non-optimal-access) MSR codes we will employ all $s+1$ kernel maps.
		\textbf{3)} Negative signs in \eqref{eq:cK} before off-block-diagonal entries are optional for the construction of optimal-access MSR codes in Section~\ref{sect:code1}.
		In other words, removing these negative signs does not affect any property of the codes.
		However, all negative signs in \eqref{eq:cK} before off-block-diagonal entries are required in Section~\ref{sect:code2}.
		Just for consistency, we keep the negative signs in \eqref{eq:cK} for both constructions.
	\end{remark}
	
	Furthermore, for $a\in [\bar n]$, $b\in [s+1]$ and a positive integer $t$, we define the following map
	\begin{equation*}
		\cM_{a,b}^{(t)} : \Fq^{s} \to \Fq^{\ell t\times \ell},
	\end{equation*}
	which maps $x_{[s]}$ to an $\ell t \times \ell$ matrix
	\begin{equation}\label{eq:cM}
		\cM_{a,b}^{(t)}(x_{[s]}) = \bI_{\ell/s^{a+1}}\otimes \left(\bI_{s^a}\bt \cK_b^{(t)}(x_{[s]})\right).
	\end{equation}
	Recall that we treat $\cK_{b}^{(t)}(x_{[s]})$ as an $s\times s$ block matrix, whose block entries are column vectors of length $t$.
	Unless otherwise specified, we always use $\cM_{a,b}^{(t)}$ in place of $\cM_{a,b}^{(t)}(x_{[s]})$.
	
	The following lemma shows that for distinct $g,h\in [\bar n]$, the matrices $\cM_{g,b}^{(t)}$ and
	$\cM_{h,b}^{(t)}$ essentially have the same structure.
	We provide its proof in Appendix~\ref{sect:proof_permu}.
	
	\begin{lemma}\label{lemma:permutation}
		For any $g,h\in [\bar n]$,
		there exist an $\ell\times \ell$ permutation matrix $\bP_{g,h}$ such that:
		\begin{itemize}
			\item [(i)] $\bP_{g,h} = \bP_{h,g}$ and $\bP_{g,h} \bP_{g,h} = \bI_{\ell}$.
			\item [(ii)] For each $a\in [\bar n]\setminus \{g,h\}$ and each $b\in [s+1]$,
			\begin{align*}
				(\bP_{g,h}\otimes \bI_t) \cM_{a,b}^{(t)} \bP_{g,h} = \cM_{a,b}^{(t)}.
			\end{align*}
			\item [(iii)] For each $b\in [s+1]$,
			\begin{align*}
				& (\bP_{g,h}\otimes \bI_t) \cM_{g,b}^{(t)} \bP_{g,h} = \cM_{h,b}^{(t)}, \\
				& (\bP_{g,h}\otimes \bI_t) \cM_{h,b}^{(t)} \bP_{g,h} = \cM_{g,b}^{(t)}.
			\end{align*}
		\end{itemize}
	\end{lemma}
	
	For a non-empty subset $B = \{b_0, b_1, \dots, b_{t-1}\}\subseteq[s+1]$ of size $t$, where $b_0 < b_1 < \dots < b_{t-1}$, a length-$st$ vector $x_{[st]}$, and a positive integer $m\ge t$, we define the $sm\times st$ matrix
	\begin{equation}\label{eq:cKB}
		\cK_B^{(m)}(x_{[st]}) = \left[ \cK_{b_0}^{(m)}(x_{[s]})\smat \cdots \smat \cK_{b_{t-1}}^{(m)}(x_{(t-1)s+[s]})  \right],
	\end{equation}
	and for $a\in [\bar n]$, we define the $\ell m \times \ell t$ matrix
	\begin{equation*}
		\cM_{a,B}^{(m)}(x_{[st]}) = \left[ \cM_{a,b_0}^{(m)}(x_{[s]})\smat \cdots \smat \cM_{a,b_{t-1}}^{(m)}(x_{(t-1)s+[s]})  \right].
	\end{equation*}
	Similarly, unless otherwise specified, we always use $\cK_{B}^{(m)}$ in place of $\cK_{B}^{(m)}(x_{[st]})$ and $\cM_{a,B}^{(m)}$ in place of $\cM_{a,B}^{(m)}(x_{[st]})$. 
	
	\begin{remark}\label{remark:M_equiv}
		By Lemma~\ref{lemma:permutation}, there exists a permutation matrix $\bP_{a,0}$ such that
		\begin{align*}
			(\bP_{a,0}\otimes \bI_{m}) 	\cM_{a,B}^{(m)} (\bI_{t}\otimes \bP_{a,0}) =\bI_{\ell/s}\bt \cK_B^{(m)},
		\end{align*}
		where $\cK_{B}^{(m)}$ is a $1\times t$ block matrix written as in \eqref{eq:cKB}.
		According to Remark~\ref{remark:block_tensor}, we further know that $\bI_{\ell/s}\bt \cK_B^{(m)}$ is permutation equivalent to $\bI_{\ell/s}\otimes \cK_B^{(m)}$.
		In conclusion, we know that $\cM_{a,B}^{(m)}$ is permutation equivalent to $\bI_{\ell/s}\otimes \cK_B^{(m)}$.
	\end{remark}
	
	\begin{corollary}\label{corollary:invert}
		Let $t=|B|$. The matrix $\cK_B^{(t)}$ is invertible if and only if $\cM_{a,B}^{(t)}$ is invertible.
	\end{corollary}
	\begin{proof}
		From Remark~\ref{remark:M_equiv}, for $m=t$ we have
		$$\det(\cM_{a,B}^{(t)}) = c\cdot \det(\bI_{\ell/s}\otimes \cK_B^{(t)}) = c\cdot \det(\cK_B^{(t)})^{l/s}$$
		where $c\in\{-1,1\}$, and this concludes the proof.
	\end{proof}
	
	Readers may opt to skip following lemmas initially; we will make reference to them once the reader arrives at the next section.
	Lemmas~\ref{lemma:UKB}-\ref{lemma:VMaB} will be used to prove the MDS property, and Lemmas~\ref{lemma:repairA}-\ref{lemma:repairB} will be used
	to show the repair scheme.
	
	\begin{lemma}\label{lemma:UKB}
		Suppose that $B\subseteq[s+1]$ is a non-empty set of size $t$.
		If $x_{[st]}$ are $st$ distinct elements such that $\cK_B^{(t)}$ is invertible, then for any integer $m > t$, 
		there exists an $sm\times sm$ matrix $U$ such that:
		\begin{itemize}
			\item [(i)]
			\begin{align}
				U \cK_B^{(m)} = \left[\begin{array}{c}
					\cK_B^{(t)} \\
					\bO
				\end{array}\right]\label{eq:cancel_K}
			\end{align}
			where $\bO$ is the $s(m-t)\times st$ all-zero matrix.
			\item [(ii)]
			For any other field element $x\notin x_{[st]}$,
			\begin{align}
				U(\bI_{s}\otimes L^{(m)}(x)) = \left[\begin{array}{l}
					~\bI_{s}\otimes L^{(t)}(x) \\
					(\bI_{s}\otimes L^{(m-t)}(x))F(x)
				\end{array}\right]\label{eq:cancel_diag_L}
			\end{align}
			where $F(x)$ is an $s\times s$ invertible matrix.
		\end{itemize}
	\end{lemma}
	
	Employing the fact that $\cM_{a,B}^{(m)}$ is permutation equivalent to $\bI_{\ell/s}\otimes \cK_B^{(m)}$ of Remark~\ref{remark:M_equiv},
	we can generalize Lemma~\ref{lemma:UKB} to the following.
	\begin{lemma}\label{lemma:VMaB}
		Suppose that $a\in [\bar n]$, $B\subseteq[s+1]$ is a non-empty set of size $t$.
		If $x_{[st]}$ are $st$ distinct elements such that $\cM_{a,B}^{(t)}$ is invertible, then for any integer $m>t$, there exists an $\ell m\times \ell m$ matrix $V$ such that:
		\begin{itemize}
			\item [(i)]
			\begin{align*}
				V \cM_{a,B}^{(m)} = \left[\begin{array}{c}
					\cM_{a,B}^{(t)} \\
					\bO
				\end{array}\right]
			\end{align*}
			where $\bO$ is the $\ell(m-t)\times \ell t$ all-zero matrix.
			\item [(ii)] For any $g\in [\bar n]\setminus \{a\}, h\in [s+1]$, and $s$ elements $y_{[s]}$ having no common elements with $x_{[st]}$,
			\begin{align*}
				V \cM_{g,h}^{(m)}(y_{[s]}) = \left[\begin{array}{l}
					\cM_{g,h}^{(t)}(y_{[s]}) \\
					\widehat \cM_{g,h}^{(m-t)}(y_{[s]})
				\end{array}\right]
			\end{align*}
			where $\widehat{\cM}_{g,h}^{(m-t)}(y_{[s]})$ is an $\ell(m-t)\times \ell $ matrix which is column equivalent to ${\cM}_{g,h}^{(m-t)}(y_{[s]})$.
			\item [(iii)]
			For any $x\notin x_{[st]}$,
			\begin{align*}
				V(\bI_{\ell}\otimes L^{(m)}(x)) = \left[\begin{array}{l}
					~\bI_{\ell}\otimes L^{(t)}(x) \\
					(\bI_{\ell}\otimes L^{(m-t)}(x))Q
				\end{array}\right]
			\end{align*}
			where $Q$ is an $\ell\times \ell$ invertible matrix.
		\end{itemize}
	\end{lemma}
	
	The proofs of these two lemmas are left in Appendix~\ref{sect:proof_UKB} and Appendix~\ref{sect:proof_VMaB} respectively.
	
	The following are some definitions and notations which will be used in the repair of single node failure.
	For any $b\in [s]$, we use $\bm e_{b}$ to denote the $b$th row of $\bI_s$.
	In other words, we have
	\begin{equation*}
		\bI_s = \left[\begin{array}{l}
			\bm e_0 \\
			\bm e_1 \\
			\vdots  \\
			\bm e_{s-1}
		\end{array}\right].
	\end{equation*}
	Let $\bar \ell = \ell / s = s^{\bar n -1}$.
	Then, for any $a\in [\bar n]$ and $b\in [s]$, we define an $\bar \ell \times \ell$ matrix
	\begin{equation}\label{eq:Rab}
		R_{a,b} = \bI_{\ell/s^{a+1}}\otimes\bm e_b \otimes \bI_{s^a}.
	\end{equation}
	If we take $\bm e_b$ as a $1\times s$ block matrix, then $$R_{a,b} = \bI_{\ell/s^{a+1}}\otimes(\bI_{s^a} \bt \bm e_b).$$
	It is easy to verify that
	\begin{equation*}
		\sum_{z\in[s]}R_{a,z}^T R_{a,z} = \bI_{\ell}.
	\end{equation*}
	
	For any index $i\in [\ell]$, we can write it as
	\begin{equation*}
		i = \sum_{z\in [\bar n]}i_z\cdot s^z = (i_{\bar n-1}, \dots, i_1, i_0),
	\end{equation*}
	where each $i_z\in [s]$.
	If $M$ is an $\ell \times \ell$ matrix, then one can check that $R_{a,b} M$ is the submatrix obtained by extracting rows from $M$ where the row index $i$ satisfies $i_a=b$. Similarly, $M R_{a,z}^T$ is the submatrix formed by extracting columns from $M$ where the column index $j$ satisfies $j_a=z$.
	
	Moreover, we use $\bm 1$ to denote the all $1$ row vector of length $s$ over $\Fq$.
	For all $a\in [\bar n]$, let
	\begin{equation*}
		R_{a,s} = \bI_{\ell/s^{a+1}}\otimes \bm 1 \otimes \bI_{s^a}.
	\end{equation*}
	Then we have $R_{a,s} = \sum_{z\in [s]}R_{a,z}$.
	
	
	For $a\in [\bar n-1], b\in [s+1]$ and a positive integer $t$, we define the following matrix
	$$
	\bar{\cM}_{a,b}^{(t)}(x_{[s]}) =  \bI_{\bar \ell /s^{a+1}}\otimes \left(\bI_{s^a}\bt \cK_b^{(t)}(x_{[s]})\right).
	$$
	Similarly, we see $\cK_{b}^{(t)}(x_{[s]})$ as an $s\times s$ block matrix, whose block entries are column vectors of length $t$.
	
	These two lemmas will be employed in the repair process of our MSR codes.
	\begin{lemma}\label{lemma:repairA}
		For any $a\in [\bar n], b,z\in[s]$, and any positive integer $t$, we have the following property.
		\begin{itemize}
			\item [(i)]
			\begin{align*}
				(R_{a,b}\otimes\bI_t)\cM_{a,b}^{(t)} R_{a,z}^{T} = \begin{cases}
					\wn \bI_{\bar{\ell}}\otimes L^{(t)}(x_z) & \text{if~} z=b,     \\
					-\bI_{\bar{\ell}}\otimes L^{(t)}(x_z)    & \text{if~} z\neq b.
				\end{cases}
			\end{align*}
			\item [(ii)] For any $h\in [s+1]\setminus\{b\}$
			\begin{align*}
				(R_{a,b}\otimes\bI_t)\cM_{a,h}^{(t)} R_{a,z}^{T} =  \begin{cases}
					\bI_{\bar{\ell}}\otimes L^{(t)}(x_b), & \text{if~} z=b,     \\
					\bO                                   & \text{if~} z\neq b.
				\end{cases}
			\end{align*}
			\item[(iii)] For any $g\in [\bar n]\setminus\{a\}$ and $h\in [s+1]$,
			\begin{align*}
				(R_{a,b}\otimes\bI_t)\cM_{g,h}^{(t)}  R_{a,z}^{T} =  \begin{cases}
					\bar \cM_{\bar g, h}^{(t)}, & \text{if~} z=b,     \\
					\bO                         & \text{if~} z\neq b,
				\end{cases}
			\end{align*}
			where
			$$
			\bar g = \begin{cases}
				g   & \text{if~} g<a, \\
				g-1 & \text{if~} g>a.
			\end{cases}
			$$
		\end{itemize}
	\end{lemma}
	The proof of Lemma~\ref{lemma:repairA} is left in Appendix~\ref{sect:proof_repairA}.
	
	\begin{lemma}\label{lemma:repairB}
		For any $a\in [\bar n], z\in [s]$, and any positive integer $t$, we have the following property.
		\begin{itemize}
			\item [(i)]
			$$
			(R_{a,s}\otimes\bI_{t})\cM_{a,s}^{(t)}R_{a,z}^T = \bI_{\bar{\ell}}\otimes L^{(t)}(x_z).
			$$
			\item [(ii)] For any $b\in [s]$,
			$$
			(R_{a,s}\otimes\bI_{t})\cM_{a,b}^{(t)}R_{a,z}^T = \begin{cases}
				\bI_{\bar{\ell}}\otimes L^{(t)}(x_b) & \text{if~} z=b,     \\
				\bO                                  & \text{if~} z\neq b.
			\end{cases}
			$$
			\item [(iii)] For any $g\in [\bar n]\setminus\{a\}, h\in [s+1]$,
			\begin{align*}
				(R_{a,s}\otimes\bI_t)\cM_{g,h}^{(t)}R_{a,z}^{T} = \bar{M}_{\bar g,h}^{(t)},
			\end{align*}
			where $\bar g$ is defined the same as that in Lemma~\ref{lemma:repairA}.
		\end{itemize}
	\end{lemma}
 \begin{proof}	
	Since $R_{a,s} = \sum_{b\in [s]}R_{a,b}$, for any $g\in [\bar n], h\in [s+1], z\in [s]$, we have
	\begin{align*}
		(R_{a,s}\otimes\bI_t)\cM_{g,h}^{(t)} R_{a,z}^T = \sum_{b\in [s]}[(R_{a,b}\otimes\bI_t)\cM_{g,h}^{(t)} R_{a,z}^T].
	\end{align*}
	Then the results follows directly from Lemma~\ref{lemma:repairA}, and we
	omit the details. 
 \end{proof}

	\section{Optimal-access MSR code with sub-packetization \texorpdfstring{$\ell = s^{n/s}$}{}}\label{sect:code1}
	In this section, we set $\bar n = \lceil n/s \rceil $, and set the sub-packetization level $\ell = s^{\bar n}$.
	As discussed in the bottom of Section~\ref{sect:intro}, without loss of generality, we can assume that $s$ divides $n$, and thus $\bar n = n/s$.
	The codeword $(C_0, C_1, \cdots C_{n-1})$ of the $(n,k,\ell)$ array code is divided into $\bar n$ groups, each of size $s$.
	We use $a\in [\bar n], b\in [s]$ to represent the index of the group and the index of the node within its group, respectively.
	In other words, group $a$ consists of the $s$ nodes $(C_{as+b}: b\in [s])$.
	
	\subsection{Construction}\label{sect:cons1}
	
	To begin with, we select $ns$ distinct elements $\lambda_{0}$, $\lambda_1$, $\dots, \lambda_{ns-1}$ from $\Fq$.
	Next, for $a\in [\bar n], b\in [s]$ and positive integer $t$ we define the \emph{kernel matrix} $K_{a,b}^{(t)}$ by assigning these elements to the map $\cK_{b}^{(t)}$ defined in \eqref{eq:cK} in the following way:
	\begin{equation}\label{eq:Kab}
		K_{a,b}^{(t)} = \cK_b^{(t)}(\lambda_{as^2+bs+[s]}).
	\end{equation}
	Obviously, for distinct values of $g,h\in [\bar n]$, the kernel matrices $K_{g,b}^{(t)}$ and $K_{h,b}^{(t)}$ exhibit identical structures as $\cK_b^{(t)}$.
	
	Following that, for $a\in[\bar n]$, $b\in [s]$ and positive integer $t$, let
	\begin{equation}\label{eq:Mab}
		M_{a,b}^{(t)} = \cM_{a,b}^{(t)}(\lambda_{as^2+bs+[s]}). 
	\end{equation}
	Moreover, for  $a\in [\bar n]$, a non-empty subset $$B = \{b_0, b_1, \dots, b_{t-1}\}\subseteq [s]$$ where $b_0<b_1< \cdots < b_{t-1}$, and positive integer $m\ge t$,
	we define the following $\ell m\times \ell t$ matrix:
	\begin{equation}\label{eq:MaB}
		M_{a, B}^{(m)}= [M_{a, b_0}^{(m)} \smat M_{a, b_1}^{(m)} \smat \cdots \smat M_{a, b_{t-1}}^{(m)}].
	\end{equation}
	If $m=t$, then $M_{a, B}^{(m)}$ is an $\ell m\times \ell m$ square matrix, and
	we write $M_{a, B}  = M_{a, B}^{(m)}$ for convenience.
	
	To guarantee the MDS property, we not only require the $ns$ elements $\lambda_{[ns]}$ to be distinct,
	but also need that
	\begin{equation}\label{eq:localAM}
		\det(M_{a, B})\neq 0,~ a\in [\bar n],~ \emptyset\neq B \subseteq [s].
	\end{equation}
	i.e., all the square matrices $M_{a,B}$ are invertible.
	We will refer to these conditions in~\eqref{eq:localAM} as \emph{local (or in-group) constraints}.
	
	The existence of such $ns$ distinct elements $\lambda_{[ns]}$ satisfy \eqref{eq:localAM} in a field of size $O_{s}(n)$ is guaranteed by the following result.
	\begin{lemma}\label{lemma:FieldSizeA}
		If $q \ge ns + (s-1) 2^{s-2}$, then we can find $ns$ distinct elements $\lambda_0,\dots, \lambda_{ns-1}$ from $\Fq$ satisfying \eqref{eq:localAM}.
		Moreover, if $n$ is large enough compared with $s$, then we can find such elements in $\Fq$ of order $q\ge ns+1$.
	\end{lemma}
	The proof of this lemma is provided in Appendix~\ref{sect:prooflemma1}.
	
	Finally, we define the parity check submatrices for our $(n,k,\ell = s^{n/s})$ MSR codes.
	For each $a\in [\bar n]$ and $b\in [s]$, the $\ell r\times \ell$ matrix $H_{as+b}$ in \eqref{eq:parity_check} is defined as
	\begin{equation}\label{eq:pcmA}
		H_{as+b} = M_{a,b}^{(r)}.
	\end{equation}
	
	\subsection{MDS property: reduction from global constraints to local constraints}\label{sect:MDS}
	
	Recall that an $(n,k,\ell)$ array code is an MDS array code if it allows us to recover any $r=n-k$ node erasures from
	the remaining $k$ nodes.
	We write the index set of failed nodes as $\cF = \{i_1, i_2, \dots, i_r\}$.
	Then the parity check equation~\eqref{eq:parity_check} can be rewritten as
	\begin{equation*}
		\sum_{i\in\cF}H_iC_i = -\sum_{i\notin\cF}H_iC_i.
	\end{equation*}
	In order to recover the failed nodes $\{C_i, i\in\cF\}$, we require the above equation has a unique solution.
	Thus the MDS property of our new codes constructed in Section~\ref{sect:code1} follows directly from the following arguments:
	\begin{equation}\label{eq:MDS}
		\begin{aligned}
			\left[H_{i_1} \smat  H_{i_2} \smat \cdots \smat H_{i_r}\right] \text{~is invertible} \\ \text{for~all~} \cF = \{i_1, i_2, \dots, i_r\}\subseteq [n].
		\end{aligned}
	\end{equation}
	
	Now we give our main result by applying Lemma~\ref{lemma:VMaB} recursively to make the matrices in global constraints~\eqref{eq:MDS} block upper triangular.
	The MDS property of our MSR codes defined in the previous subsection will follow directly from
	the following.
	\begin{lemma}\label{lemma:globalA}
		For any $z$ distinct integers $a_0, a_1, \dots, a_{z-1}\in [\bar n]$ and any $z$ non-empty subsets $B_0,B_1,\dots, B_{z-1}\subseteq [s]$ satisfying
		$|B_0| + |B_1| + \cdots + |B_{z-1}| = m \le r$, we have
		\begin{equation*}
			\det\left([ M_{a_0, B_0}^{(m)}\smat M_{a_1, B_1}^{(m)} \smat \cdots  \smat M_{a_{z-1}, B_{z-1}}^{(m)} ]\right)\neq0.
		\end{equation*}
	\end{lemma}
	\begin{proof}
		We prove it by induction on the positive integer $z$.
		If $z = 1$, since all the $\lambda_i$s satisfying local constraints~\eqref{eq:localAM}, we have $\det(M_{a_0, B_0})\neq0$.
		
		For the inductive hypothesis, we assume that the conclusion holds for an arbitrary positive integer $z$.
		
		For the case of $z+1$, we have $|B_0|+|B_1|+\cdots +|B_z| = m$. We write $t = |B_0|$. As $M_{a_0, B_0}$ is invertible and all the $ns$ elements $\lambda_{ns}$ are distinct, according to Lemma~\ref{lemma:VMaB}, there exist $\ell m\times \ell m$ matrix $V$ such that
		\begin{align}\label{eq:det}
			\begin{split}
				&\det(V [M_{a_0, B_0}^{(m)} \smat M_{a_1, B_1}^{(m)} \smat \cdots \smat M_{a_{z},B_{z}}^{(m)} ])\\[.5em]
				=&
				\det\left[\begin{array}{c c c c}
					M_{a_0,B_0}^{(t)} & M_{a_1,B_1}^{(t)}            & \cdots & M_{a_z,B_z}^{(t)}            \\[.5em]
					\bO               & \widehat M_{a_1,B_1}^{(m-t)} & \cdots & \widehat M_{a_z,B_z}^{(m-t)}
				\end{array} \hspace*{-0.05in}\right]\\[.5em]
				=& \det\left( M^{(t)}_{a_0, B_0}\right) \det\left([\widehat M_{a_1,B_1}^{(m-t)}\smat  \cdots\smat  \widehat M_{a_z,B_z}^{(m-t)}]\right)
			\end{split}
		\end{align}
		and each $\widehat M_{a_i,B_i}^{(m-t)}$ is column equivalent to $M_{a_i,B_i}^{(m-t)}$.
		Note that $|B_1|+\cdots +|B_z| = m-t$.
		By the induction hypothesis, the matrix  $$[M_{a_1,B_1}^{(m-t)} \smat M_{a_2,B_2}^{(m-t)} \smat  \cdots\smat  M_{a_z,B_z}^{(m-t)}]$$ is invertible.
		It means that the matrix $$[\widehat M_{a_1,B_1}^{(m-t)}\smat  \widehat M_{a_2,B_2}^{(m-t)} \smat  \cdots\smat  \widehat M_{a_z,B_z}^{(m-t)}]$$ is also invertible.
		By~\eqref{eq:det}, we know $$\det([ M_{a_0, B_0}^{(m)}\smat M_{a_1, B_1}^{(m)}\smat \dots \smat M_{a_{z}, B_{z}}^{(m)} ])\neq 0.$$
	\end{proof}
	
	Note that if we set  $m = r$ in Lemma~\ref{lemma:globalA}, then we obtain global constraints \eqref{eq:MDS} directly.

	\subsection{Repair scheme of \texorpdfstring{$(n,k,s^{n/s})$}{} MSR code}\label{sect:repair1}
	
	We denote the index set of helper nodes as $\cH$.
	For every $a\in [\bar n]$ and $b\in [s]$, the node $C_{as+b}$ can be recovered from $$\{R_{a,b}C_{j}:j\in \cH\}$$ for any set
	$\cH\subseteq[n]\setminus\{as+b\}$ with size $|\cH| = d$.
	Where $R_{a,b}$ is defined in \eqref{eq:Rab}.
	
	Since each row of the matrix $R_{a,b}$ has only one nonzero element and all the helper nodes use the same repair matrix $R_{a,b}$ to download helper symbols, this repair procedure has both the optimal access property and constant repair property.
	
	{\bf How to repair $C_{as+b}$,  $a\in [\bar n],b\in [s]$.}
	Recall that the parity check equation of the code $C$ is
	$$
	H_0C_0 + \dots + H_1C_1 = \bzero.
	$$
	Each $H_i$ in \eqref{eq:parity_check} has $\ell$ block rows, and each block row includes $r$ rows.
	In other words, we have $\ell$ sets of parity check equations and each set of them has $r$ parity check equations.
	The repair of any failed node $C_{as+b}$ only involves $\bar\ell = \ell/s$ out of these $\ell$ sets of parity check equations.
	To be specific, we use the repair matrix $R_{a,b}$ to choose the $\bar{\ell}$ block rows of each $H_i$ whose block row indices $i$ satisfy $i_a=b$.
	More precisely, we have
	\begin{align*}
		& (R_{a,b}\otimes \bI_{r})(\sum_{i\in [n]}H_iC_i)\nonumber                                               \\
		= & \sum_{i\in [n]} (R_{a,b}\otimes\bI_{r}) H_i(\sum_{z\in [s]}R_{a,z}^T R_{a,z})C_i\nonumber              \\
		= & \sum_{i\in [n]}\sum_{z\in [s]} \big[(R_{a,b}\otimes\bI_{r}) H_i  R_{a,z}^T\big](R_{a,z}C_i)\nonumber   \\
		= & \sum_{i\in [n]\setminus\{as+b\}} \big[(R_{a,b}\otimes\bI_{r}) H_i  R_{a,b}^T\big](R_{a,b}C_i)\nonumber \\ &+ \sum_{z\in [s]} \big[(R_{a,b}\otimes\bI_{r}) H_{as+b}  R_{a,z}^T\big](R_{a,z}C_{as+b}) \\=& \bzero \nonumber
	\end{align*}
	from \eqref{eq:parity_check}, \eqref{eq:cK}, \eqref{eq:cM}, \eqref{eq:Mab}, \eqref{eq:pcmA} and Lemma~\ref{lemma:repairA}.
	For any $i\in [n]\setminus\{as+b\}$, we set
	$$\bar H_i = (R_{a,b}\otimes\bI_{r}) H_i R_{a,b}^T,\quad \bar C_i = R_{a,b}C_i,$$
	and for any $z\in [s]$, we set
	$$\bar H_{as+b}^{\left\langle z\right\rangle} = (R_{a,b}\otimes\bI_{r}) H_{as+b} R_{a,z}^T,~~ \bar C_{as+b}^{\left\langle z\right\rangle} = R_{a,z}C_{as+b}.$$
	Then the repair equation of the failed node $C_{as+b}$ can be written as
	\begin{equation}\label{eq:pcrA}
		\sum_{i\in[n]\setminus\{as+b\}}\bar H_i\bar C_i + \sum_{z\in[s]}\bar{H}_{as+b}^{\left\langle z\right\rangle}\bar{C}_{as+b}^{\left\langle z\right\rangle}  =\bzero.
	\end{equation}
	By Lemma~\ref{lemma:repairA}, we know that all these $n-1+s$ matrices $\bar{H}_i, i\in [n]\setminus \{as+b\}$ and $\bar H_{as+b}^{\left\langle z\right\rangle }, z\in [s]$ are $\bar{\ell}\times \bar{\ell}$ block matrices whose entries are column vectors of length $r$,
	and $\bar{C}_i, i\in [n]\setminus \{as+b\}$, $C_{as+b}^{\left\langle z\right\rangle }, z\in [s]$ are length-$\bar{\ell}$ column vectors.
	
	To be specific, let $L_i^{(r)} = L^{(r)}(\lambda_i)$, we have
	\begin{itemize}
		\item[(i)] for $z\in [s]$,
		\begin{align*}
			\bar H_{as+b}^{ \left\langle z\right\rangle } =  \begin{cases}
				\wn \bI_{\bar{\ell}}\otimes L_{as^2+bs+b}^{(r)} & \text{if~} z=b,       \\[.3em]
				-\bI_{\bar{\ell}}\otimes L_{as^2+bs+z}^{(r)}    & \text{if~} z\neq  b ;
			\end{cases}
		\end{align*}
		\item[(ii)] for $h\in [s]\setminus\{b\}$,
		\begin{equation*}
			\bar H_{as+h} =  \bI_{\bar{\ell}}\otimes L_{as^2+hs+b}^{(r)};
		\end{equation*}
		\item[(iii)] and for $g\in [\bar n]\setminus\{a\}$, $h\in [s]$,
		\begin{align*}
			\bar{H}_{gs+h} = \bar \cM_{\bar g,h}^{(r)}(\lambda_{gs^2+hs+[s]}),
		\end{align*}
		where $\bar g$ is defined as in Lemma~\ref{lemma:repairA}.
	\end{itemize}
	
	We observe that \eqref{eq:pcrA} define an array code of length $n-1+s$.
	The nodes of this new array code can also be divided into $\bar n$ groups as follows.
	For $g\in [\bar n]\setminus\{a\}$, the group $g$ consists of the $s$ nodes $(\bar C_{gs+h}, h\in [s])$.
	The group $a$ consists of the $2s-1$ nodes $\bar C_{as+b}^{\left\langle z\right\rangle }, z\in [s]$, and $\bar C_{as+h}, h\in [s]\setminus\{b\}$.
	Moreover, the parity check submatrices $\bar{H}_{gs+h}$, $g\in [\bar n]\setminus\{a\}, h\in [s]$ are precisely the $n-s$ parity check submatrices that would appear in the MSR code construction with code length $n-s$ and sub-packetization $\bar{\ell}$.
	The parity check submatrices of group $a$, $\bar{H}_{as+b}^{\left\langle z\right\rangle }, z\in [s]$ and $\bar{H}_{as+h}, h\in [s]\setminus\{b\}$ are block-diagonal matrices,
	and the diagonal entries are the same within each matrix.
	Since the $\lambda_i$'s that appear in $\bar{H}_{as+b}^{\left\langle z\right\rangle }, z\in [s]$ and $\bar{H}_{as+h}, h\in [s]\setminus\{b\}$
	are distinct from the $\lambda_i$'s that appear in  $\bar{H}_{gs+h}, g\in [\bar n]\setminus\{as+b\}, h\in [s]$.
	The method we used to prove the MDS property of $(n,k,\ell = s^{n/s})$ MSR code construction in Section~\ref{sect:MDS}, together with Lemma~\ref{lemma:VMaB}~(iii), can be generalized to show that \eqref{eq:pcrA} also
	defines an $(n+s-1,k+s-1,\bar \ell)$ MDS array code
	$$(\bar{C}_0,\dots, \bar{C}_{as+b-1}, \bar{C}_{as+b}^{\left\langle 0\right\rangle }, \dots, \bar{C}_{as+b}^{\left\langle s-1\right\rangle },  \bar{C}_{as+b+1},\dots, \bar{C}_{n-1}).$$
	Therefore, $(\bar{C}_{as+b}^{\left\langle z\right\rangle }, z\in [s])$ can be recovered from any $d = k+s-1$ vectors in the set $(\bar{C}_i, i\in[n]\setminus\{as+b\})$.
	When we know the values of $(\bar{C}_{as+b}^{\left\langle z\right\rangle }, z\in [s])$, we are able to recover all the coordinates of $C_{as+b}$.
	
	\subsection{A toy example}
	
	In this subsection, we give an example of $(n$, $k$, $\ell = s^{n/s})$ MSR code with parameters $n = 6, k = 2$ and $d = 4$ over $\ff_{32}$.
	Then $r = n-k = 4, s = d-k+1 = 3$ and $\ell =s^{n/s} = 9$.
	Let $\alpha$ be a root of the primitive polynomial $x^5 + x^2 +1$ over $\ff_{32}$, and we set $\lambda_i=\alpha^{i}$ for $i\in [18]$,
	and
	\begin{equation*}
		L_i=L^{(4)}(\lambda_i)  = \left[
		\begin{array}{l}
			1           \\
			\alpha^i    \\
			\alpha^{2i} \\
			\alpha^{3i} \\
		\end{array}
		\right].
	\end{equation*}
	For any $a\in [2], b\in [3]$ and positive integer $t$,
	we set $K_{a,b}^{(t)}$ and $M_{a,b}^{(t)}$ as in \eqref{eq:Kab} and \eqref{eq:Mab}, respectively.
	According to \eqref{eq:pcmA}, the parity check submatrices $H_{3a+b} = M_{a,b}^{(4)}$.
	It can be verified that these elements $\lambda_{[18]}$ satisfy the local constraints \eqref{eq:localAM}.
	To explicitly present the code construction, the sequence of parity check submatrices $H_0, H_1, \dots, H_{5}$ are listed as follows:
	{\scriptsize
		\begin{align*}
			& H_0 =
			\left[\begin{array}{@{\hspace*{-0.02in}}*{2}{*{3}{@{\hspace*{-0.02in}}c}|}*{3}{@{\hspace*{-0.02in}}c}@{\hspace*{-0.02in}}}
				\rL{0}  & \rL[-]{1} & \rL[-]{2} & \gbzero & \gbzero   & \gbzero   & \gbzero & \gbzero   & \gbzero   \\
				\gbzero & \rL{1}    & \gbzero   & \gbzero & \gbzero   & \gbzero   & \gbzero & \gbzero   & \gbzero   \\
				\gbzero & \gbzero   & \rL{2}    & \gbzero & \gbzero   & \gbzero   & \gbzero & \gbzero   & \gbzero   \\
				\hline
				\gbzero & \gbzero   & \gbzero   & \rL{0}  & \rL[-]{1} & \rL[-]{2} & \gbzero & \gbzero   & \gbzero   \\
				\gbzero & \gbzero   & \gbzero   & \gbzero & \rL{1}    & \gbzero   & \gbzero & \gbzero   & \gbzero   \\
				\gbzero & \gbzero   & \gbzero   & \gbzero & \gbzero   & \rL{2}    & \gbzero & \gbzero   & \gbzero   \\
				\hline
				\gbzero & \gbzero   & \gbzero   & \gbzero & \gbzero   & \gbzero   & \rL{0}  & \rL[-]{1} & \rL[-]{2} \\
				\gbzero & \gbzero   & \gbzero   & \gbzero & \gbzero   & \gbzero   & \gbzero & \rL{1}    & \gbzero   \\
				\gbzero & \gbzero   & \gbzero   & \gbzero & \gbzero   & \gbzero   & \gbzero & \gbzero   & \rL{2}
			\end{array}\right],  \\
			& H_1 =
			\left[\begin{array}{@{\hspace*{-0.02in}}*{2}{*{3}{@{\hspace*{-0.02in}}c}|}*{3}{@{\hspace*{-0.02in}}c}@{\hspace*{-0.02in}}}
				\rL{3}    & \gbzero & \gbzero   & \gbzero   & \gbzero & \gbzero   & \gbzero   & \gbzero & \gbzero   \\
				\rL[-]{3} & \rL{4}  & \rL[-]{5} & \gbzero   & \gbzero & \gbzero   & \gbzero   & \gbzero & \gbzero   \\
				\gbzero   & \gbzero & \rL{5}    & \gbzero   & \gbzero & \gbzero   & \gbzero   & \gbzero & \gbzero   \\
				\hline
				\gbzero   & \gbzero & \gbzero   & \rL{3}    & \gbzero & \gbzero   & \gbzero   & \gbzero & \gbzero   \\
				\gbzero   & \gbzero & \gbzero   & \rL[-]{3} & \rL{4}  & \rL[-]{5} & \gbzero   & \gbzero & \gbzero   \\
				\gbzero   & \gbzero & \gbzero   & \gbzero   & \gbzero & \rL{5}    & \gbzero   & \gbzero & \gbzero   \\
				\hline
				\gbzero   & \gbzero & \gbzero   & \gbzero   & \gbzero & \gbzero   & \rL{3}    & \gbzero & \gbzero   \\
				\gbzero   & \gbzero & \gbzero   & \gbzero   & \gbzero & \gbzero   & \rL[-]{3} & \rL{4}  & \rL[-]{5} \\
				\gbzero   & \gbzero & \gbzero   & \gbzero   & \gbzero & \gbzero   & \gbzero   & \gbzero & \rL{5}
			\end{array}\right], \\
			& H_2 =
			\left[\begin{array}{@{\hspace*{-0.02in}}*{2}{*{3}{@{\hspace*{-0.02in}}c}|}*{3}{@{\hspace*{-0.02in}}c}@{\hspace*{-0.02in}}}
				\rL{6}    & \gbzero   & \gbzero & \gbzero   & \gbzero   & \gbzero & \gbzero   & \gbzero   & \gbzero \\
				\gbzero   & \rL{7}    & \gbzero & \gbzero   & \gbzero   & \gbzero & \gbzero   & \gbzero   & \gbzero \\
				\rL[-]{6} & \rL[-]{7} & \rL{8}  & \gbzero   & \gbzero   & \gbzero & \gbzero   & \gbzero   & \gbzero \\
				\hline
				\gbzero   & \gbzero   & \gbzero & \rL{6}    & \gbzero   & \gbzero & \gbzero   & \gbzero   & \gbzero \\
				\gbzero   & \gbzero   & \gbzero & \gbzero   & \rL{7}    & \gbzero & \gbzero   & \gbzero   & \gbzero \\
				\gbzero   & \gbzero   & \gbzero & \rL[-]{6} & \rL[-]{7} & \rL{8}  & \gbzero   & \gbzero   & \gbzero \\
				\hline
				\gbzero   & \gbzero   & \gbzero & \gbzero   & \gbzero   & \gbzero & \rL{6}    & \gbzero   & \gbzero \\
				\gbzero   & \gbzero   & \gbzero & \gbzero   & \gbzero   & \gbzero & \gbzero   & \rL{7}    & \gbzero \\
				\gbzero   & \gbzero   & \gbzero & \gbzero   & \gbzero   & \gbzero & \rL[-]{6} & \rL[-]{7} & \rL{8}
			\end{array}\right], \\
			& H_3 =
			\left[\begin{array}{@{\hspace*{-0.02in}}*{2}{*{3}{@{\hspace*{-0.02in}}c}|}*{3}{@{\hspace*{-0.02in}}c}@{\hspace*{-0.02in}}}
				\rL{9}  & \gbzero & \gbzero & \rL[-]{10} & \gbzero    & \gbzero    & \rL[-]{11} & \gbzero    & \gbzero    \\
				\gbzero & \rL{9}  & \gbzero & \gbzero    & \rL[-]{10} & \gbzero    & \gbzero    & \rL[-]{11} & \gbzero    \\
				\gbzero & \gbzero & \rL{9}  & \gbzero    & \gbzero    & \rL[-]{10} & \gbzero    & \gbzero    & \rL[-]{11} \\
				\hline
				\gbzero & \gbzero & \gbzero & \rL{10}    & \gbzero    & \gbzero    & \gbzero    & \gbzero    & \gbzero    \\
				\gbzero & \gbzero & \gbzero & \gbzero    & \rL{10}    & \gbzero    & \gbzero    & \gbzero    & \gbzero    \\
				\gbzero & \gbzero & \gbzero & \gbzero    & \gbzero    & \rL{10}    & \gbzero    & \gbzero    & \gbzero    \\
				\hline
				\gbzero & \gbzero & \gbzero & \gbzero    & \gbzero    & \gbzero    & \rL{11}    & \gbzero    & \gbzero    \\
				\gbzero & \gbzero & \gbzero & \gbzero    & \gbzero    & \gbzero    & \gbzero    & \rL{11}    & \gbzero    \\
				\gbzero & \gbzero & \gbzero & \gbzero    & \gbzero    & \gbzero    & \gbzero    & \gbzero    & \rL{11}
			\end{array}\right], \\
			& H_4 =
			\left[\begin{array}{@{\hspace*{-0.02in}}*{2}{*{3}{@{\hspace*{-0.02in}}c}|}*{3}{@{\hspace*{-0.02in}}c}@{\hspace*{-0.02in}}}
				\rL{12}    & \gbzero    & \gbzero    & \gbzero & \gbzero & \gbzero & \gbzero    & \gbzero    & \gbzero    \\
				\gbzero    & \rL{12}    & \gbzero    & \gbzero & \gbzero & \gbzero & \gbzero    & \gbzero    & \gbzero    \\
				\gbzero    & \gbzero    & \rL{12}    & \gbzero & \gbzero & \gbzero & \gbzero    & \gbzero    & \gbzero    \\
				\hline
				\rL[-]{12} & \gbzero    & \gbzero    & \rL{13} & \gbzero & \gbzero & \rL[-]{14} & \gbzero    & \gbzero    \\
				\gbzero    & \rL[-]{12} & \gbzero    & \gbzero & \rL{13} & \gbzero & \gbzero    & \rL[-]{14} & \gbzero    \\
				\gbzero    & \gbzero    & \rL[-]{12} & \gbzero & \gbzero & \rL{13} & \gbzero    & \gbzero    & \rL[-]{14} \\
				\hline
				\gbzero    & \gbzero    & \gbzero    & \gbzero & \gbzero & \gbzero & \rL{14}    & \gbzero    & \gbzero    \\
				\gbzero    & \gbzero    & \gbzero    & \gbzero & \gbzero & \gbzero & \gbzero    & \rL{14}    & \gbzero    \\
				\gbzero    & \gbzero    & \gbzero    & \gbzero & \gbzero & \gbzero & \gbzero    & \gbzero    & \rL{14}
			\end{array}\right], \\
			& H_5 =
			\left[\begin{array}{@{\hspace*{-0.02in}}*{2}{*{3}{@{\hspace*{-0.02in}}c}|}*{3}{@{\hspace*{-0.02in}}c}@{\hspace*{-0.02in}}}
				\rL{15}    & \gbzero    & \gbzero    & \gbzero    & \gbzero    & \gbzero    & \gbzero & \gbzero & \gbzero \\
				\gbzero    & \rL{15}    & \gbzero    & \gbzero    & \gbzero    & \gbzero    & \gbzero & \gbzero & \gbzero \\
				\gbzero    & \gbzero    & \rL{15}    & \gbzero    & \gbzero    & \gbzero    & \gbzero & \gbzero & \gbzero \\
				\hline
				\gbzero    & \gbzero    & \gbzero    & \rL{16}    & \gbzero    & \gbzero    & \gbzero & \gbzero & \gbzero \\
				\gbzero    & \gbzero    & \gbzero    & \gbzero    & \rL{16}    & \gbzero    & \gbzero & \gbzero & \gbzero \\
				\gbzero    & \gbzero    & \gbzero    & \gbzero    & \gbzero    & \rL{16}    & \gbzero & \gbzero & \gbzero \\
				\hline
				\rL[-]{15} & \gbzero    & \gbzero    & \rL[-]{16} & \gbzero    & \gbzero    & \rL{17} & \gbzero & \gbzero \\
				\gbzero    & \rL[-]{15} & \gbzero    & \gbzero    & \rL[-]{16} & \gbzero    & \gbzero & \rL{17} & \gbzero \\
				\gbzero    & \gbzero    & \rL[-]{15} & \gbzero    & \gbzero    & \rL[-]{16} & \gbzero & \gbzero & \rL{17}
			\end{array}\right],
		\end{align*}
	}where $\bzero$ is the all-zero vector of length $4$.

	\section{Non-optimal-access MSR code with smaller sub-packetization \texorpdfstring{$\ell = s^{n/(s+1)}$}{}}\label{sect:code2}
	
	In this section, we suppose that $(s+1)\mid n$ and set the sub-packetization $\ell = s^{\bar n}$, where $\bar n = n/(s+1)$.
	The codeword $(C_0, C_1, \cdots C_{n-1})$ is divided into $\bar n$ groups of size $s+1$.
	(Note that in the previous construction, the group size is $s$.)
	We use $a\in [\bar n], b\in [s+1]$ to denote the index of the group and the index of the node within its group, i.e.,
	for each $a\in [\bar n]$, the $s+1$ nodes $(C_{a(s+1)+b}: b\in [s+1])$ are in group $a$.
	
	\subsection{Construction}
	
	If we choose $ns$ distinct elements $\lambda_{[ns]}$ from the finite field $\Fq$, then the kernel matrices $K_{a,b}^{(t)}$ for $a\in [\bar n], b\in [s+1]$ and a positive integer $t$ can be defined as
	\begin{equation}\label{eq:Kab2}
		K_{a,b}^{(t)} = \cK_{b}^{(t)}(\lambda_{as(s+1)+bs+[s]}).
	\end{equation}
	And we have the $\ell t\times \ell$ matrices
	\begin{equation}\label{eq:Mab2}
		M_{a,b}^{(t)} = \cM_{a,b}^{(t)}(\lambda_{as(s+1)+bs+[s]}).
	\end{equation}
	
	For $a\in [\bar n]$, a nonempty subset $B\subseteq[s+1]$ of size $t$, and a positive integer $m \ge t$,
	we define the $\ell m\times \ell t$ matrices $M_{a,B}^{(m)}$ in the same way as \eqref{eq:MaB}.
	Similarly, we omit the superscript if $m = t$.
	
	The construction in this subsection also requires that the $ns$ distinct $\lambda_i$s satisfy the
	following local constraints
	\begin{equation} \label{eq:localB}
		\det(M_{a, B})\not=0,~~ a \in [\bar n],~~ \emptyset\neq B \subseteq [s+1],
	\end{equation}
	i.e., all the square matrices $M_{a, B}$ are invertible.
	The existence of such $\lambda_i$s in a field of size $O_{s}(n)$ is guaranteed by the following result.
	\begin{lemma}\label{lemma:FieldSizeB}
		If a prime power $q \ge ns + s 2^{s-1}$, then we can find $ns$ distinct elements $\lambda_0, \dots, \lambda_{ns-1}$ from $\Fq$ satisfying \eqref{eq:localB}.
		Moreover, if $n$ is large enough compared to $s$, then we can find these elements in $\Fq$ of order $q\ge ns+1$.
	\end{lemma}
	We omit the proof of Lemma~\ref{lemma:FieldSizeB} as it is similar to that of Lemma~\ref{lemma:FieldSizeA}.
	
	Now we are ready to define the parity check submatrices for the $(n,k,\ell = s^{n/(s+1)})$ MSR code.
	For each $a\in [\bar n]$ and $b\in [s+1]$,
	the $r\ell\times \ell$ matrix $H_{a(s+1)+b}$ in \eqref{eq:parity_check} is defined as
	\begin{equation}\label{eq:pcmB}
		H_{a(s+1)+b} = M_{a,b}^{(r)}.
	\end{equation}
	
	\subsection{MDS property}
	The MDS property of our MSR codes defined by \eqref{eq:parity_check}, \eqref{eq:cK}, \eqref{eq:cM}, \eqref{eq:Kab2}, \eqref{eq:Mab2}, \eqref{eq:localB} and \eqref{eq:pcmB} will follow directly from
	the following global constraints.
	\begin{lemma}\label{lemma:globalB}
		For any $z$ distinct integers $a_0, a_1, \dots, a_{z-1} \in [\bar n]$ and $z$ subsets $B_0,B_1,\dots, B_{z-1}\subseteq [s+1]$ satisfying
		$|B_0| + |B_1| + \cdots + |B_{z-1}| = m\le r$, we have
		\begin{equation*}
			\det\left([ M_{a_0, B_0}^{(m)}\smat M_{a_1, B_1}^{(m)}\smat \dots \smat M_{a_{z-1}, B_{z-1}}^{(m)} ]\right)\not=0.
		\end{equation*}
	\end{lemma}
	Since the proof of Lemma~\ref{lemma:globalB} is almost the same as that of Lemma~\ref{lemma:globalA},
	we omit it.
	
	\subsection{ Repair scheme of \texorpdfstring{$(n,k,s^{n/(s+1)})$}{} MSR codes.}
	We still denote the index set of helper nodes as $\cH$.
	We divide the discussion into two cases:
	\begin{itemize}
		\item \textbf{Repairing the failed node $C_{a(s+1)+b}$, for $a\in [\bar n]$ and $b\in [s]$.}
		For any set $\cH\subseteq[n]\setminus\{a(s+1)+b\}$ with size $|\cH| = d$, the node $C_{a(s+1)+b}$ can be recovered from $$\{R_{a,b}C_i : i\in \cH\}$$.
		\item \textbf{Repairing the failed node $C_{a(s+1)+s}$, for $a\in [\bar n]$.}
		For any set $\cH\subseteq[n]\setminus\{a(s+1)+s\}$ with size $|\cH| = d$,
		we set
		$$
		\cH_1 = \cH \cap \big(a(s+1)+[s+1]\big),
		$$
		and set
		$$
		\cH_2 = \cH \setminus \big(a(s+1)+[s+1]\big).
		$$
		Then the last node $C_{a(s+1)+s}$ in group $a$ can be recovered from
		$$\{R_{a,z}C_{a(s+1)+z}: a(s+1)+z\in \cH_1\}$$
		and
		$$\{R_{a,s}C_i : i\in \cH_2\}.$$
	\end{itemize}
	
	From the above repair scheme, we can see that
	the repair procedure of failed node $C_{a(s+1)+b}$ for $b\in [s]$ also has the optimal access property and constant repair property.
	But the repair procedure of the last node of each group does not have these two properties.
	
	Next, we illustrate the repair procedure.
	The procedure to repair the first $s$ nodes of each group is exactly the same as the repair procedure of $(n,k,\ell = s^{(n/s)})$ MSR code described in Section~\ref{sect:repair1}, so we omit the detail.
	
	{\bf How to repair the last node $C_{a(s+1)+s}$ of group $a$.}
	The $\ell$ parity check equations involved in $H_i, i\in [n]$ are divided into $\bar \ell = \ell/s$ groups, each of $s$ equations.
	The equation indices $i$ within the same group only differ in the $a$th digit of their expansion in base $s$.
	In order to repair $C_{s}$, we sum up each group of $s$ block rows of all matrices in \eqref{eq:parity_check}.
	To be specific, we have
	\begin{align}
		& (R_{a,s}\otimes \bI_{r})(\sum_{i\in [n]}H_iC_i)\nonumber                                                      \\
		= & \sum_{i\in [n]} (R_{a,s}\otimes\bI_{r}) H_i(\sum_{z\in [s]}R_{a,z}^T R_{a,z})C_i\nonumber                     \\
		= & \sum_{i\in [n]}\sum_{z\in [s]} \big[(R_{a,s}\otimes\bI_{r}) H_i  R_{a,z}^T\big](R_{a,z}C_i)\nonumber          \\
		= & \sum_{z\in [s]}\big[(R_{a,s}\otimes\bI_{t})H_{a(s+1)+s}R_{a,z}^T\big](R_{a,z}C_{a(s+1)+s})\nonumber           \\
		& +\sum_{b\in [s]}\big[(R_{a,s}\otimes\bI_{t})H_{a(s+1)+b}R_{a,b}^T\big](R_{a,b}C_{a(s+1)+b})\nonumber          \\
		& +\sum_{i\in [n]\setminus(a(s+1)+[s+1])}\big[(R_{a,s}\otimes\bI_{t})H_{i}R_{a,0}^T\big](R_{a,s}C_{i})\nonumber \\
		= & \bzero.\nonumber
	\end{align}
	from \eqref{eq:parity_check}, \eqref{eq:cK}, \eqref{eq:cM}, \eqref{eq:Mab2}, \eqref{eq:pcmB} and Lemma~\ref{lemma:repairB}.
	For any $z\in [s]$, we set
	\begin{align*}
		\bar H_{a(s+1)+s}^{\left\langle z\right\rangle } = & (R_{a,s}\otimes\bI_{t})H_{a(s+1)+s}R_{a,z}^T, \\
		\bar C_{a(s+1)+s}^{\left\langle z\right\rangle } = & R_{a,z}C_{a(s+1)+s}.
	\end{align*}
	For any $b\in [s]$, we set
	\begin{align*}
		\bar H_{a(s+1)+b} = & (R_{a,s}\otimes\bI_{t})H_{a(s+1)+b}R_{a,b}^T, \\
		\bar C_{a(s+1)+b} = & R_{a,b}C_{a(s+1)+b}.
	\end{align*}
	For any $i\in [n]\setminus(a(s+1)+[s+1])$, we set
	\begin{align*}
		\bar H_{i} = & (R_{a,s}\otimes\bI_{t})H_{i}R_{a,0}^T, \\
		\bar C_{i} = & R_{a,s}C_{i}.
	\end{align*}
	Then the repair equation of the failed node $C_{a(s+1)+b}$ can be written as
	\begin{equation}\label{eq:pcrB}
		\sum_{i\in[n]\setminus\{a(s+1)+s\}}\bar H_i\bar C_i + \sum_{z\in[s]}\bar{H}_{a(s+1)+s}^{\left\langle z\right\rangle}\bar{C}_{a(s+1)+s}^{\left\langle z\right\rangle}  =\bzero.
	\end{equation}
	By Lemma~\ref{lemma:repairB}, we know that all these $n-1+s$ matrices $\bar{H}_{i}, i\in [n]\setminus\{a(s+1)+s\}$, $\bar H_{a(s+1)+s}^{\left\langle z\right\rangle}, z\in [s]$ are $\bar{\ell}\times \bar{\ell}$ block matrices whose entries are column vectors of length $r$,
	and $\bar{C}_i, i\in [n]\setminus \{a(s+1)+s\}$, $C_{a(s+1)+s}^{\left\langle z\right\rangle }, z\in [s]$ are length-$\bar{\ell}$ column vectors.
	
	To be specific, let $L_i^{(r)} = L^{(r)}({\lambda_i})$, we have
	\begin{itemize}
		\item[(i)] for any $z\in [s]$
		$$\bar H_{a(s+1)+s}^{ \left\langle z\right\rangle } =   \bI_{\bar{\ell}} \otimes L_{as(s+1)+s^2+z}^{(r)};$$
		\item[(ii)] for $b\in [s]$,
		\begin{equation*}
			\bar H_{a(s+1)+b} =  \bI_{\bar{\ell}}\otimes L_{as(s+1)+bs+b}^{(r)};
		\end{equation*}
		\item[(iii)] and for $g\in [\bar n]\setminus\{a\}$, $h\in [s+1]$,
		\begin{align*}
			\bar{H}_{g(s+1)+h} = \bar \cM_{\bar g,h}^{(r)}(\lambda_{gs(s+1)+hs+[s]}),
		\end{align*}
		where $\bar g$ is defined as in Lemma~\ref{lemma:repairB}.
	\end{itemize}
	One can see that the matrices $\bar H_{gs+h}, g\in[\bar n]\setminus \{a\}, h\in [s+1]$ are precisely the $n-s-1$ parity check submatrices that would appear in the
	MSR code construction with code length $n-s-1$ and sub-packetization $\bar{\ell}$.
	The other matrices $\bar{H}_{a(s+1)+b}, b\in [s]$ and $\bar{H}_{a(s+1)+s}^{\left\langle z\right\rangle }, z\in [s]$ are block-diagonal matrices,
	and the diagonal entries are the same within each matrix.
	Moreover, the $\lambda_i$'s that appear in $\bar{H}_{a(s+1)+b}, b\in [s]$ and $\bar{H}_{a(s+1)+s}^{\left\langle z\right\rangle }, z\in [s]$
	do not intersect with the $\lambda_i$'s that appear in  $\bar{H}_{i}, i\in[n]\setminus([a(s+1)+[s+1])$.
	The method we used to prove the MDS property of $(n,k,\ell = s^{n/(s+1)})$ MSR code construction can be generalized to show that \eqref{eq:pcrB} also defines an
	$(n+s-1,k+s-1,\bar \ell)$ MDS array code
	$$(\bar{C}_0,\dots, \bar{C}_{m-1},\quad \bar{C}_{m}^{\left\langle 0\right\rangle }, \dots, \bar{C}_{m}^{\left\langle s-1\right\rangle },\quad \bar{C}_{m+1},\dots, \bar{C}_{n-1}),$$ where $m = a(s+1)+s$.
	Therefore, $(\bar{C}_m^{\left\langle z\right\rangle }, z\in [s])$ can be recovered from any $d = k+s-1$ vectors in the set $(\bar{C}_z, z\in[n]\setminus\{m\})$.
	When we know the values of $(\bar{C}_m^{\left\langle z\right\rangle }, z\in [s])$, we are able to recover all the coordinates of $C_m$.
	
	\appendix
	
	\section{Proof of Lemma~\ref{lemma:permutation}}\label{sect:proof_permu}
	In this section, we write $L_i^{(t)} = L^{(t)}(x_i)$.
	For any index $i\in [\ell]$, we can write it as
	\begin{equation}\label{eq:base_s}
		i = \sum_{z\in [\bar n]}i_z\cdot s^z = (i_{\bar n-1}, \dots, i_1, i_0),
	\end{equation}
	where each $i_z\in [s]$.
	Now we treat $\cM_{a,b}^{(t)}$ as an $\ell \times \ell$ block matrix, whose block entries are coloumn vectors of length $t$.
	Using the $\bar{n}$-digit base-$s$ expansion of index $i,j\in [\ell]$,
	we can verify that the block entry of $\cM_{a,b}^{(t)}$ at the cross of the $i$th block row and the $j$th (block) row is
	\begin{equation}\label{eq:cM_base_s}
		\cM_{a,b}^{(t)}(i, j) =
		\begin{cases}
			\wn L_{j_a}^{(t)} & \text{if~} i = j,                                  \\
			-L_{j_a}^{(t)}    & \text{if~} i_a=b\neq j_a   \text{~and}             \\
			& \ i_z=j_z\text{~for~} z\in [\bar n]\setminus\{a\}, \\
			\wn \bzero        & \text{otherwise.}
		\end{cases}
	\end{equation}
	Next, we define a permutation
	$\pi_{g,h} : [\ell]\to [\ell]$ by
	\begin{align*}
		\pi_{g,h}(i)  = & i_gs^h + i_hs^g + \sum_{z\in [\bar n]\setminus\{g,h\}} i_zs^z = \sum_{z\in [\bar n]}\tilde i_zs^z = \tilde i,
	\end{align*}
	i.e., $\tilde i$ is obtained from exchanging values of $i_g$ and $i_h$ in base-$s$ expansion of $i$.
	Hence $\pi_{g,h}=\pi_{h,g}$ and $\pi_{g,h} \circ \pi_{g,h}$ is the identity permutation.
	We denote the permutation matrix associated to $\pi_{g,h}$ as $\bP_{g,h}$.
	Then the $i$-th row of $\bP_{g,h}$ is exactly the $\tilde i$th row of $\bI_{\ell}$, i.e., for $i,j\in [\ell],$
	$$\bP_{g,h}(i,j) = \bI_{\ell}(\pi_{g,h}(i), j) = \bI_{\ell}(i, \pi_{g,h}(j)).$$
	
	By the properties of $\pi_{g,h}$, we directly have $\bP_{g,h}=\bP_{h,g}$ and $\bP_{g,h} \bP_{g,h} = \bI_{\ell}.$
	For $a\in[\bar n]$, $b\in[s+1]$ and a positive integer $t$, we define
	$$\widetilde \cM_{a,b}^{(t)} = (\bP_{g,h}\otimes \bI_t)\cM_{a,b}^{(t)}\bP_{g,h}.$$
	We can check that
	$$\widetilde{\cM}_{a,b}^{(t)}(i,j) = \cM_{a,b}^{(t)}(\pi_{g,h}(i), \pi_{g,h}(j)) = \cM_{a,b}^{(t)}(\tilde{i}, \tilde{j}).$$
	Here we also treat $\widetilde{\cM}_{a,b}^{(t)}$ as an $\ell \times \ell$ block matrix, whose block entries are coloumn vectors of length $t$.
	For $b\in [s+1]$, $i,j\in[\ell]$, by \eqref{eq:cM_base_s} we have
	\begin{equation*}
		\begin{aligned}
			\cM_{g,b}^{(t)}(\tilde{i}, \tilde{j})=
			\begin{cases}
				\wn L_{\tilde j_g}^{(t)} & \text{if~} \tilde i = \tilde j,                                 \\
				-L_{\tilde j_g}^{(t)}    & \text{if~} \tilde i_g=b\neq \tilde j_g  \text{~and}             \\
				& \ \tilde i_z=\tilde j_z\text{~for~}z\in [\bar n]\setminus\{g\}, \\
				\wn \bzero               & \text{otherwise,}
			\end{cases} \\
			\cM_{h,b}^{(t)}(i,j)=
			\begin{cases}
				\wn L_{j_h}^{(t)} & \text{if~}   i =   j,                              \\
				-L_{j_h}^{(t)}    & \text{if~} \ i_h=b\neq   j_h   \text{~and}         \\
				& \ i_z= j_z\text{~for~}z\in [\bar n]\setminus\{h\}, \\
				\wn \bzero        & \text{otherwise.}
			\end{cases}
		\end{aligned}
	\end{equation*}
	One can check that
	$$
	\widetilde\cM_{g,b}^{(t)}(i,j) = \cM_{g,b}^{(t)}(\tilde{i}, \tilde{j}) = \cM_{h,b}^{(t)}(i,j),
	$$
	thus $\widetilde \cM_{g,b}^{(t)}= \cM_{h,b}^{(t)}.$
	By the same way, we can verify that $\widetilde \cM_{h,b}^{(t)}= \cM_{g,b}^{(t)}$ and for $a\in[\bar n]\setminus\{g,h\}$, $\widetilde \cM_{a,b}^{(t)} = \cM_{a,b}^{(t)}$.

	\section{Proof of Lemma~\ref{lemma:UKB}}\label{sect:proof_UKB}
	
	In this section, we write $K = \cK_B^{(m)}$, which is an $sm\times st$ matrix.
	We use $K(i,\rc), i\in [sm]$ to represent the $i$th row of $K$, $K(\rc, j), j\in [st]$ represent the $j$th column of $K$, and $K(i,j)$ denote its entry at the cross of the $i$th row and the $j$th column.
	\begin{itemize}
		\item [(i)]
		Since
		$\cK_{B}^{(t)}$ is invertible, the $st$ rows
		$$
		K(um+v,\rc), u\in [s], v\in [t]
		$$
		are linear independent.
		Thus we can write each row in $K$ as a linear combination of these $st$ rows.
		To begin with,
		for each $z\in [s]$, we have
		\begin{align}\label{eq:cancel_t}
			K(zm+t,\rc)  + \sum_{u\in [s], v\in [t]} c_{z,u,v}K(um+v,\rc) = \bzero,
		\end{align}
		where $c_{z,u,v}$ are the coefficient of the linear combination.
		The equation~\eqref{eq:cancel_t} means that we use the rows in $\cK_{B}^{(t)}$ to cancel the $(zm+t)$th row in $K$.
		
		Back to the structure of $K$, for each $z\in [s]$, $i\in [t]$, $j\in [s]$, we know the following column vector
		\begin{align*}\left[
			\begin{array}{ll}
				K(zm,    is+j)      \\
				K(zm+1,  is+j)      \\
				\qquad\qquad \vdots \\
				K(zm+m-1,is+j)
			\end{array}\right]
		\end{align*}
		is either $\pm L^{(m)}(x_{is+j})$ or $\bzero$.
		Thus, for each $t<w<m$, we have $$K(zm+w,is+j) = x_{is+j}^{w-t}K(zm+t, is+j).$$
		
		Combining the above with~\eqref{eq:cancel_t}, it implies that
		\begin{align}\label{eq:cancel_w}
			K(zm+w, \rc)  + \hspace*{-1.5em}\sum_{u\in [s], v\in [t]}\hspace*{-1em} c_{z,u,v}K(um+v+w-t,\rc) = \bzero.
		\end{align}
		
		Given these coefficients $c_{z,u,v}, z,u\in[s], v\in [t]$, we construct the following two matrices.
		Firstly, we define an $s\times sm$ matrix $U^*$ as follows. For each $z,u\in [s]$ and $v\in [m]$,
		\begin{align}
			U^*(z,um+v) = \begin{cases}
				c_{z,u,v} & \text{if~} v\in [t], \\
				1         & \text{if~} u=z, v=t, \\
				0         & \text{otherwise.}
			\end{cases}\label{eq:U1}
		\end{align}
		From \eqref{eq:cancel_t}, we know that
		\begin{equation}\label{eq:U1KO}
			U^*  \cdot K = \bO.
		\end{equation}
		
		Next, we define an $sm\times sm$ matrix $\bar U$ as follows. For each $z,u\in [s], w, v\in [m]$,
		\begin{align*}
			\bar U(zm+w,um+v) = \begin{cases}
				1             & \hspace*{-4em}\text{if~} z=u, w=v, \\
				c_{z,u,v-w+t} &                                    \\&\hspace*{-4em}\text{if~} w\in [m]\setminus[t], \\&\hspace*{-4em}\text{and~} v\in w-t+[t],\\
				0             & \hspace*{-4em}\text{otherwise.}
			\end{cases}
		\end{align*}
		The matrix $\bar U$ can be seen as an $s\times s$ block matrix, where each block is an $m\times m$ matrix.
		The structure of the block entry of $\bar U$ at the $z$th block row and the $u$th block column is shown in Fig.~\ref{fig:barU}.
		We can compute that if $w\in [t]$, then $$(\bar UK)(zm+w, \rc) = K(zm+w, \rc).$$
		Otherwise, $w\in [m]\setminus[t]$, from \eqref{eq:cancel_w}, $$(\bar UK)(zm+w, \rc) = \bzero.$$
		\begin{figure}
			\centering
			\includegraphics[scale = 0.65]{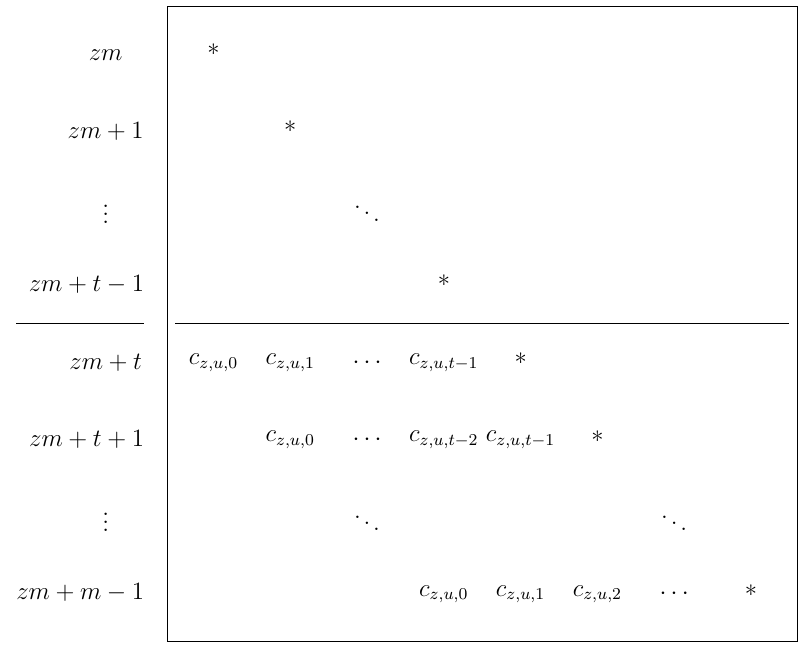}
			\caption{The scalar entries in the $z$th block row and the $u$th block column of the matrix $\bar U$, where each block of $\bar U$ has the same size $m\times m$.
				Please note that if $z=u$, the ``*'' sign is $1$.
				Otherwise, $z\neq u$, the ``*'' sign is $0$.}\label{fig:barU}
		\end{figure}
		Let
		\begin{align*}
			E = \left[\begin{array}{l}
				\bI_s\otimes \left[\begin{array}{cc} \bI_t & \hspace*{.05in}\bO\end{array}\hspace*{.135in}\right] \\[1ex]
				\bI_s\otimes \left[\begin{array}{cc} \bO   & \bI_{m-t}\end{array}\right]
			\end{array}\right],
		\end{align*}
		and let $U = E\bar U$.
		We can check that
		\begin{align*}
			U K =
			\left[\begin{array}{c}
				\cK_{B}^{(t)} \\
				\bO
			\end{array} \hspace*{-0.05in}\right],
		\end{align*}
		where the function of matrix $E$ is to move the last $m-t$ rows of each block row of $\bar UK $, in this case they are all-zero rows, to the bottom.
		
		\item [(ii)]
		Firstly, we define an $s\times s$ matrix
		\begin{equation*}
			F(x) = U^*\cdot(\bI_s\otimes L^{(m)}(x)) = (f_{z,u}(x), z,u\in [s]).
		\end{equation*}
		By  definition \eqref{eq:U1}, we can check that
		\begin{align}
			f_{z,u}(x) = & \sum_{v\in [m]} U^*(z,um+v) x^{v}\nonumber                                  \\
			=            & \begin{cases}
				\sum_{v\in [t]} c_{z,u,v} x^{v} + x^{t} & \text{if~} z=u,   \\
				\sum_{v\in [t]} c_{z,u,v} x^{v}         & \text{otherwise.}
			\end{cases}\label{eq:fzu}
		\end{align}
		This implies that $f_{z,u}$ is a nonzero polynomial in $\Fq[x]$ of degree $t$ if $z=u$, while $f_{z,u}$ is a polynomial of degree less than $t$ if $z\neq u$.
		
		Next, we analyze the element of the $sm\times s$ matrix $A = \bar U(\bI_s\otimes L^{(m)}(x))$.
		For the entry at the cross $(zm+w)$th row and $u$th column, $z,u\in [s], w\in [m]$, we have
		\begin{align*}
			& A(zm+w,u)=\sum_{v\in [m]} \bar U(zm+w, um+v) x^{v}                                                                                                                    \\
			= & \begin{cases}
				x^{w}                                         & \hspace*{-5em}\text{if~} z = u, w\in [t], \\
				\sum_{v'\in [t]} c_{z,u,v'}x^{v'+w-t} + x^{w} &                                           \\& \hspace*{-5em} \text{if~} z=u,w\in [m]\setminus[t],     \\
				\sum_{v'\in [t]} c_{z,u,v'}x^{v'+w-t}         &                                           \\& \hspace*{-5em}\text{if~} z\neq u, w\in [m]\setminus[t], \\
				0                                             & \hspace*{-5em}\text{otherwise,}
			\end{cases}
		\end{align*}
		where $v' = v-w+t$.
		Combining with \eqref{eq:fzu}, we know that
		\begin{align*}
			A(zm+w,u)=
			\begin{cases}
				x^{w}                     & \text{if~} z = u, w\in [t],      \\
				f_{z,u}(\lambda_i)x^{w-t} & \text{if~} w\in [m]\setminus[t], \\
				0                         & \text{otherwise.}
			\end{cases}
		\end{align*}
		Thus, we know that
		$$E A = U\cdot(\bI_s\otimes L^{(m)}(x)) = \left[\begin{array}{c}
			\bI_{s}\otimes L_i^{(t)} \\
			F(x)\otimes L_i^{(m-t)}
		\end{array} \hspace*{-0.05in}\right].$$
		
		The last thing we need to show is that $F(x)$ is invertible.
		This is equivalent to showing $f(x) = \det(F(x))\neq 0$.
		From \eqref{eq:fzu}, we know that $f$ is a nonzero polynomial of degree $st$.
		Due to $x\notin x_{[st]}$, it is sufficient to show that these $st$ elements $x_{[st]}$ are exactly the $st$ distinct roots of $f$.
		
		Suppose that $B = \{b_0, \dots, b_{t-1}\}$, where $b_0 < \dots, b_{t-1}$.
		According to~\eqref{eq:cK}, \eqref{eq:cKB}, by comparing the columns of the matrix $K$ and $\bI_s\otimes L^{(m)}(x_{is+j})$,
		we observe that for any $i\in [t]$ and $j\in [s]$:
		\begin{itemize}
			\item[(1)]  if $j=b_i$, then $$K(\rc,is+j)=(\bI_s\otimes L^{(m)}(x_{is+j}))(\cdot,j);$$
			\item[(2)]  if $j\neq b_i$, then
			\begin{align*}
				K(\rc,is+j)= & (\bI_s\otimes L^{(m)}(x_{is+j}))(\rc,j) \\&-(\bI_s\otimes L^{(m)}(x_{is+j}))(\rc ,b_i).
			\end{align*}
		\end{itemize}
		From~\eqref{eq:U1KO} we have $U^*\cdot K(\rc,is+j)=\bzero$.
		Combining this with above observation, we see that
		the columns of $F(x_{is+j})=U^* \cdot (\bI_s\otimes L^{(m)}(x_{is+j}))$ are linearly dependent.
		Therefore, the determinant $f(x_{is+j}) = 0$ for $i\in [t], j\in [s]$,
		and $x_{[st]}$ are exactly the $st$ roots of $f$.
		This implies that for each $x\notin x_{st}$, $F(x)$ is invertible.
	\end{itemize}
	\section{Proof of Lemma~\ref{lemma:VMaB}}\label{sect:proof_VMaB}
	
	We will use Lemma~\ref{lemma:UKB} to show our results for the special case of $a=0$.
	Following that, we use Lemma~\ref{lemma:permutation} to prove all general cases.
	
	Recall that $\cM_{0,B}^{(m)} = \bI_{\ell/s}\bt \cK_B^{(m)}$, where we see $\cK_{B}^{(m)}$ as a $1\times t$ block matrix as written in \eqref{eq:cKB}.
	For convenience, we write $\cK_h^{(m)} = \cK_h^{(m)}(y_{[s]})$ and $\cM_{g,h}^{(m)} = \cM_{g,h}^{(m)}(y_{[s]})$.
	
	According to Corollary~\ref{corollary:invert}, we know that $\cK_B^{(t)}$ is also invertible.
	By Lemma~\ref{lemma:UKB}, there exists an $sm\times sm$ matrix $U$ satisfying \eqref{eq:cancel_K} and \eqref{eq:cancel_diag_L}.
	Let
	\begin{equation*}
		G =
		\left[\begin{array}{l}
			\bI_{\ell/s}\otimes \left[\begin{array}{cc} \bI_{st} & \bO\end{array}\hspace*{.23in}\right] \\[1ex]
			\bI_{\ell/s}\otimes \left[\begin{array}{cc} \bO   & \bI_{s(m-t)}\end{array}\right]
		\end{array}\right],
	\end{equation*}
	and let $ V = G (\bI_{\ell/s}\otimes U)$.
	We can verify that when $a=0$,
	such $\ell m\times \ell m$ matrix $V$ satisfies all the properties we require.
	
	\begin{itemize}
		\item [(i)] According to \eqref{eq:cancel_K}, we know that
		\begin{equation}\label{eq:iea}
			\begin{aligned}
				& (\bI_{\ell/s}\otimes U)  \cM_{0,B}^{(m)}                \\
				= & (\bI_{\ell/s}\otimes U)   (\bI_{\ell/s}\bt \cK_B^{(m)}) \\
				= & \bI_{\ell/s}\bt \left[ U \cK_B^{(m)}\right]             \\
				= & \bI_{\ell/s}\bt \left [
				\begin{array}{c}
					\cK_{B}^{(t)} \\
					\bO
				\end{array}
				\right],
			\end{aligned}
		\end{equation}
		where $\left[U \cK_{B}^{(m)}\right]$ and $\left [
		\begin{array}{c}
			\cK_{B}^{(t)} \\
			\bO
		\end{array}
		\right]$ are regarded as $1\times t$ block matrices with block entries of the same size.
		Then we have
		\begin{align*}
			V  \cM_{0,B}^{(m)} =
			\left[\begin{array}{c}
				\cM_{0,B}^{(t)} \\
				\bO
			\end{array} \right].
		\end{align*}
		Here we use the permutation matrix
		$G$
		to move the all-zero rows of \eqref{eq:iea} to the bottom.
		
		\item[(ii)]
		Since
		\begin{align}
			\cM_{g,h}^{(m)} =         & \bI_{\ell/s^{g+1}}\otimes(\bI_{s^g}\bt \cK_{h}^{(m)}) \nonumber                      \\
			=                         & \bI_{\ell/s^{g+1}}\otimes(\bI_{s^{g-1}}\bt (\bI_{s}\bt \cK_{h}^{(m)})),\label{eq:36} \\
			(\bI_{\ell/s}\otimes U) = & \bI_{\ell/s^{g+1}}\otimes(\bI_{s^{g-1}}\bt (\bI_{s}\otimes U))\label{eq:37},
		\end{align}
		and
		\begin{align}\label{eq:38}
			G =\hspace*{-.5em}
			\left[\hspace*{-.5em}\begin{array}{l}
				\bI_{\ell/s^{g+1}}\otimes \Big(\bI_{s^{g-1}}\bt \big(\bI_s\otimes \left[\bI_{st}\smat  \hspace*{1.05em}\bO\hspace*{1.05em}\right]\big)\Big) \\[1ex]
				\bI_{\ell/s^{g+1}}\otimes \Big(\bI_{s^{g-1}}\bt \big(\bI_s\otimes \left[\bO \smat   \bI_{s(m-t)}\right]\big)\Big)
			\end{array}\hspace*{-.8em}\right],
		\end{align}
		where $\bI_s\bt \cK_{h}^{(m)}$ in  \eqref{eq:36}, $\bI_{s}\otimes U$ in \eqref{eq:37}, $\bI_s\otimes [\begin{array}{cc} \bI_{st} & \bO\end{array}]$ and $\bI_s\otimes [\begin{array}{cc} \bO   & \bI_{s(m-t)}\end{array}]$ in \eqref{eq:38} are uniformly partitioned into $s\times s$ block matrices.
		We have
		\begin{align*}
			V \cM_{g,h}^{(m)}
			=\hspace*{-.5em}
			\left[\hspace*{-.5em}\begin{array}{l}
				\bI_{\ell/s^{g+1}}\otimes \big(\bI_{s^{g-1}}\bt K'\big) \\[1ex]
				\bI_{\ell/s^{g+1}}\otimes \big(\bI_{s^{g-1}}\bt K''\big)
			\end{array}\hspace*{-.8em}\right],
		\end{align*}
		where
		\begin{align*}
			K'  = \big(\bI_s\otimes \left[\bI_{st}\smat  \hspace*{1.05em}\bO\hspace*{1.05em}\right]\big) (\bI_s\otimes U) (\bI_s \bt \cK_{h}^{(m)}), \\
			K'' = \big(\bI_s\otimes \left[\bO \smat   \bI_{s(m-t)}\right]\big) (\bI_s\otimes U) (\bI_s \bt \cK_{h}^{(m)}).
		\end{align*}
		Let $K = (\bI_s\otimes U) (\bI_s \bt \cK_{h}^{(m)})$.
		If we regard $K$ as an $s\times s$ block matrix with block entries of the same size, then by \eqref{eq:cK}, and \eqref{eq:cancel_diag_L} in Lemma~\ref{lemma:UKB},
		we have that for $i,j\in [s]$
		$$
		K(i,j) = \begin{cases}
			\wn\left[\begin{array}{c} \bI_{s}\otimes L^{(t)}(y_j) \\ F(y_{j})\otimes L^{(m-t)}(y_j) \end{array}\right] & \text{if~} i = j,        \\[1.5em]
			-\left[\begin{array}{c} \bI_{s}\otimes L^{(t)}(y_j) \\ F(y_{j})\otimes L^{(m-t)}(y_j) \end{array}\right]   & \text{if~} i = h \neq j, \\[1.5em]
			\wn \bI_{s} \otimes \bzero                                                                                 & \text{otherwise,}
		\end{cases}
		$$
		where $F(y_{j})$ is an $s\times s$ invertible matrix for all $j\in [s]$.
		Now we get that,
		\begin{align*}
			K' =  \bI_{s}\bt \cK_{h}^{(t)},\quad K''= (\bI_{s}\bt \cK_{h}^{(m-t)}) P'
		\end{align*}
		where $P' = \diag(F(y_{j}) : j\in [s])$ is a $s^2\times s^2$ invertible matrix.
		Thus we get that
		\begin{align*}
			V \cM_{g,h}^{(m)}
			=\hspace*{-.5em}
			\left[\hspace*{-.5em}\begin{array}{l}
				\cM_{g,h}^{(t)} \\[1ex]
				\cM_{g,h}^{(m-t)} P
			\end{array}\right],
		\end{align*}
		where $P = \bI_{\ell/s^{g+1}} \otimes (\bI_{\ell/s^{g-1}}\bt P')$.
		It is easy to check that $Q$ is also an invertible matrix.
		If we write $\widehat{\cM}_{g,h}^{(m-t)} = \cM_{g,h}^{(m-t)} P$, then we know that $\widehat{\cM}_{g,h}^{(m-t)}$ is equivalent to ${\cM}_{g,h}^{(m-t)}$.
		
		\item [(iii)] Similar as the analysis in $(ii)$, we know that
		\begin{align*}
			V (\bI_{\ell}\otimes L^{(m)}(x)) =
			\left[\begin{array}{c}
				\bI_{\ell}\otimes L^{(t)}(x) \\
				(\bI_{\ell} \otimes L^{(m-t)}(x))(\bI_{\ell/s}\otimes F(x))
			\end{array} \hspace*{-0.05in}\right]
		\end{align*}
		where $F(x)$ is an $s\times s$ invertible matrix.
	\end{itemize}
	
	For the case of $a\neq 0$, by Lemma~\ref{lemma:permutation}, we know that there exists a permutation matrix $\bP_{a,0}$ such that
	\begin{align*}
		(\bP_{a,0}\otimes \bI_{m})\cM_{a,B}^{(m)}(\bI_{t}\otimes \bP_{a,0}) = \bI_{\ell/s}\bt\cK_{B}^{(m)}.
	\end{align*}
	and
	\begin{align*}
		(\bP_{a,0}\otimes\bI_m)\cM_{g,h}^{(m)} \bP_{a,0} = \bI_{\ell/s^{z+1}}\otimes(\bI_{s^z}\bt \cK_{h}^{(m)}),
	\end{align*}
	where
	$$
	0<z = \begin{cases}
		a & \text{if~} g=0,    \\
		g & \text{if~} g\neq0.
	\end{cases}
	$$
	In addition, we know that $$
	(\bP_{a,0}\otimes\bI_m)(\bI_\ell\otimes L^{(m)}(x)) \bP_{a,0} = \bI_\ell\otimes L^{(m)}(x).
	$$
	Then one can check that the following matrix
	\begin{equation*}
		 \left [\begin{array}{c|c}
			\bP_{a,0}\otimes \bI_{t} & \bO                        \\
			\hline
			\bO                      & P_{a,0}\otimes \bI_{(m-t)}
		\end{array}\right] V (\bP_{a,0}\otimes \bI_{m}).
	\end{equation*}
	satisfies all the three conditions in this lemma.

	\section{Proof of Lemma~\ref{lemma:repairA}}\label{sect:proof_repairA}
	In this section, we write $L_i^{(t)} = L^{(t)}(x_i)$ for convenience.
	In light of the $\bar n$-digit base-$s$ expansion of the index $i\in [\ell]$ given in \eqref{eq:base_s}, it becomes apparent that multiplying any matrix with $\ell$ rows on the left by $R_{a,b}$ is tantamount to extracting a specific subset of $\bar \ell = \ell/s$ rows.
	These rows are characterized by the condition that each index $i$ satisfies $i_a = b$.
	For each $a\in [\bar n]$ and $b\in [s]$ we define the function $$\tau_{a,b} : [\bar{\ell}] \to [\ell]$$ by
	\begin{equation*}
		\begin{aligned}
			\tau_{a,b}(i) = & \sum_{z\in [a]}i_zs^z + bs^{a} + \sum_{z\in [\bar n - 1]\setminus[a]}i_zs^{z+1}.
		\end{aligned}
	\end{equation*}
	Recall that $R_{a,b}$ defined in \eqref{eq:Rab} is an $\bar{\ell}\times \ell$ matrices.
	We can check that for $i\in [\bar{\ell}]$ and $j\in [\ell]$,
	\begin{equation*}
		R_{a,b}(i,j) = \begin{cases}
			1 & \text{if~} j = \tau_{a,b}(i), \\
			0 & \text{otherwise.}
		\end{cases}
	\end{equation*}
	
	For any $\ell \times \ell$ block matrix $M$ with block entries of size $t\times 1$, we know that
	$\bar M = (R_{a,b} \otimes \bI_t) M R_{a,z}^T$
	is an $\bar{\ell} \times \bar{\ell}$ block matrix.
	Furthermore, $\bar M$ is obtained by choosing those block rows and columns whose indices satisfy $i_a = b$ and $j_a = z$, where $i$ and $j$ are the index of block row and (block) column, respectively.
	
	For any $g\in [\bar n], h\in [s+1], z\in [s]$ and $i,j\in [\bar{\ell}]$, we have
	\begin{align*}
		[(R_{a,b}\otimes\bI_t)\cM_{g,h}^{(t)} R_{a,z}^T](i,j) = \cM_{g,h}^{(t)}(\tau_{a,b}(i), \tau_{a,z}(j)).
	\end{align*}
	Further, using the  $\bar n$-digit base-$s$ expansion of index \eqref{eq:base_s} and \eqref{eq:cM_base_s}, we have
	\begin{align}
		& \cM_{g,h}^{(t)}(\tau_{a,b}(i), \tau_{a,z}(j))\nonumber                                                                                    \\
		= & \begin{cases}
			\wn L_{(\tau_{a,z}(j))_g}^{(t)} & \text{if~} \tau_{a,b}(i) = \tau_{a,z}(j),                                           \\
			-L_{(\tau_{a,z}(j))_g}^{(t)}    & \text{if~} (\tau_{a,b}(i))_g = h \neq (\tau_{a,z}(j))_g, \text{~and}                \\
			& \hspace*{-1.8em }(\tau_{a,b}(i))_m =(\tau_{a,z}(j))_m, m\in [\bar n]\setminus\{g\}, \\
			\wn \bzero                      & \text{otherwise.}
		\end{cases}\label{eq:smallM}
	\end{align}
	
	In the following we only prove the part $(iii)$ of this lemma, i.e., the case of $g\neq a, h\in [s+1]$.
	From \eqref{eq:smallM} we can check that
	\begin{align*}
		\cM_{g,h}^{(t)}(\tau_{a,b}(i), \tau_{a,b}(j)) =
		\begin{cases}
			\wn L_{j_{\bar g}}^{(t)} & \text{if~} i = j,                                             \\
			-L_{j_{\bar g}}^{(t)}    & \text{if~} i_{\bar g}=h\neq j_{\bar g}   \text{~and}          \\
			& \hspace*{-.5in}i_m=j_m, m\in [\bar n-1]\setminus\{{\bar g}\}, \\
			\wn \bzero               & \text{otherwise.}
		\end{cases}
	\end{align*}
	Similarly as \eqref{eq:cM_base_s}, we can get
	$$(R_{a,b}\otimes\bI_t)\cM_{g,h}^{(t)} R_{a,b}^T = \bI_{\bar{\ell}/s^{\bar g+1}}\otimes (\bI_{s^{\bar g}}\bt \cK_h^{(t)}).$$
	Moreover, for $z\in [s]\setminus\{b\}$, by \eqref{eq:smallM} we have
	$$
	\cM_{g,h}^{(t)}(\tau_{a,b}(i), \tau_{a,z}(j)) = \bzero,
	$$
	i.e., $(R_{a,b}\otimes\bI_t)\cM_{g,h}^{(t)} R_{a,z}^T = \bO$.
	The arguments $(i)$ and $(ii)$ can be proved similarly as above.

	\section{Proof of Lemma~\ref{lemma:FieldSizeA}}\label{sect:prooflemma1}
	In this section, we write
	$\cK_{B} = \cK_{B}^{(|B|)}(x_{sB+[s]})$ and $\cM_{a,B} = \cM_{a,B}^{(|B|)}(x_{sB+[s]})$,
	where $sB+[s]:=\{sb+j: b\in B, j\in [s]\}$.
	Recall that the $ns$ distinct elements $\lambda_{[ns]}$ are required to make the square matrix $M_{a,B} = \cM_{a,B}(\lambda_{as^2+sB+[s]})$ invertible for all $a\in [\bar n]$ and non-empty subset $B\subseteq [s]$.
	By Corollary~\ref{corollary:invert} we have $\cM_{a,B}$ is invertible if and only if $\cK_{B}$ is invertible.
	Thus the local constraints~\eqref{eq:localAM} can be rewritten as
	\begin{equation}\label{eq:localA}
		\det(\cK_{B}(\lambda_{as^2+sB+[s]}))\neq 0, a\in [\bar n], \emptyset\neq B\subseteq[s],
	\end{equation}
	
	Following that, we define the function
	\begin{equation*}
		D(x_{[s^2]}) = \prod_{B\subseteq [s]} \det\big(\cK_{B}\big),
	\end{equation*}
	Then the local constraints \eqref{eq:localAM} and \eqref{eq:localA} is equivalent to that
	\begin{equation*}
		D(\lambda_{as^2+[s^2]})\neq 0, a\in [\bar n].
	\end{equation*}
	
	We will use the following well-known result in our proof.
	\begin{lemma}[Combinatorial Nullstellensatz, see, N. Alon \cite{alon_1999}]\label{lemma:alon}
		Let $F$ be an arbitrary field, and let $f = f(x_1, \dots, x_n)$ be a  polynomial in $F[x_1, \dots, x_n]$.
		Suppose the degree $\deg(f)$ of $f$ is $\sum_{i = 1}^{n}t_i$, where each $t_i$ is a nonnegative integer, and suppose the coefficient of $\prod_{i=1}^{n}x_i^{t_i}$ in $f$ is nonzero.
		Then, if $S_1,\cdots,S_n$ are subsets of $F$ with $|S_i| > t_i$, there are  $s_1\in S_1, s_2\in S_2, \dots, s_n\in S_n$ so that
		$$f(s_1, \dots, s_n)\neq0.$$
	\end{lemma}
	
	\begin{corollary}\label{corollary:1}
		Let $F$ be an arbitrary field, and let $f = f(x_1, \dots, x_n)$ be a  polynomial in $F[x_1, \dots, x_n]$.
		Suppose the degree $\deg(f)$ of $f$ is $\sum_{i = 1}^{n}t_i$, where each $t_i$ is a nonnegative integer,
		and suppose the coefficient of $\prod_{i=1}^{n}x_i^{t_i}$ in $f$ is nonzero.
		Let $S\subseteq F$ be a subset with size $|S| > \max(t_1, \dots, t_n)$.
		Then for any $i\in\{1,2,\dots,n\}$ and any $T\subseteq S$ of size $|T|>t_i$, there exists $s \in T$ such that
		$$f(x_1, \dots, x_{i-1}, s, x_{i+1} \dots, x_n)$$
		is a nonzero polynomial in $F[x_1, \dots, x_{i-1},x_{i+1}, \dots, x_n]$.
	\end{corollary}
	\begin{proof}
		As $|S|>\max(t_1,\dots, t_n)$, for $1\leq j\leq n, j\neq i$, we can choose a subset $S_j\subseteq S$ with $S_j>t_j$.
		According to Lemma~\ref{lemma:alon}, there are $s_j\in S_j,1\leq j\leq n, j\neq i$ and $s\in T$ such that
		$f(s_1,\dots,s_{i-1},s,s_{i+1} \dots, s_n)\neq0.$
		Thus $f(x_1, \dots, x_{i-1}, s, x_{i+1}, \dots, x_n)$ is a nonzero polynomial with $n-1$ variables.
	\end{proof}
	
	\begin{corollary}\label{corollary:2}
		Let $F$ be an arbitrary field, and let $f = f(x_1, \dots, x_n)$ be a nonzero polynomial in $F[x_1, \dots, x_n]$.
		Let $d_i$ be the largest power of $x_i$ occurring in any monomial of $f$, and $d = \max(d_1, \dots, d_n)$.
		Let $S$ be a subset of $F$ with $|S| \ge n + d$.
		Then there exist  $n$ distinct elements $s_1, s_2,\dots, s_n\in S$ such that
		$$f(s_1, \dots, s_n)\neq0.$$
	\end{corollary}
	
	\begin{proof}
		We can find $s_1,s_2,\dots,s_n$ as follows.
		\begin{itemize}
			\item[1)]
			
			Let $S^{(1)} = S$.
			Suppose the degree of $f(x_1,\dots,x_n)$ is $\sum_{i=1}^{n}t_i$,
			and the coefficient of $\prod_{i=1}^{n}x_i^{t_i}$ in $f(x_1,\dots,x_n)$ is nonzero.
			Since $|S^{(1)}|\ge d+n> \max(t_1,\dots,t_n)$, by Corollary~\ref{corollary:1} we can find $s_1\in S^{(1)}$ such that $f(s_1,x_2,\dots,x_n)$ is a nonzero polynomial.
			
			\item [2)]
			
			Let $S^{(2)} = S \setminus\{s_1\}$.
			Suppose the degree of $f(s_1$, $x_2,\dots,x_n)$ is $\sum_{i=2}^{n}t_i'$, and the coefficient of $\prod_{i=2}^{n}x_i^{t_i'}$ in $f(s_1,x_2,\dots,x_n)$ is nonzero.
			Since $|S^{(2)}|\ge d+n-1> \max(t_2',\dots,t_n')$, we can find $s_2\in S^{(2)}$ such that $f(s_1,s_2,x_3,\dots,x_n)$ is a nonzero polynomial.
		\end{itemize}
		We can repeat this process for $s_3,\dots,s_n$, and this concludes the proof.
	\end{proof}
	
	Now we give the proof of Lemma~\ref{lemma:FieldSizeA}.
	We first prove that the function $D(x_{[s^2]})$ is a nonzero polynomial in $\Fq[x_{[s^2]}]$, and its degree of each variable $x_i, i\in [s^2]$ is $(s-1)2^{s-2}$.
	Then by Corollary~\ref{corollary:2} we can choose $\lambda_{[s^2]}$ from a subset $S\subseteq \Fq$ with size $|S| \ge s^2+(s-1) 2^{s-2}$ satisfying that $D(\lambda_{[s^2]})\neq0$.
	
	Suppose that $B = \{b_0,\dots, b_{t-1}\}$ and $b_0 < \dots< b_{t-1}$.
	We regard $\det(\cK_{B})$ as a polynomial function of variables $x_{bs+j}, b\in B , j\in [s]$.
	Next, we show that $\det(\cK_{B})$ is a nonzero polynomial and $\deg(\det(\cK_{B})) = \frac{st(t-1)}{2}$.
	To see this,  recall that $\cK_{B}$ is an $st\times st$ square matrix, and nonzero entries on the same row have the same degree.
	To be specific, for $i\in [s], j\in [t]$,  all the nonzero entries in the $it+j$th row of $\cK_{B}$ have the same degree $j$.
	Therefore, $\det(\cK_{B})$ is a homogeneous polynomial.
	Note that the term
	\begin{equation}\label{eq:KBterm}
		\prod_{j\in [s]}1\cdot \prod_{j\in [s]} x_{b_1s+j}^1 \cdot \prod_{j\in [s]} x_{b_2s+j}^2  \cdots \prod_{j\in [s]} x_{b_{t-1}s+j}^{t-1}
	\end{equation}
	appears in the expansion of $\det(\cK_B)$, and can only be obtained by
	picking constant entries from $\cK_{b_0}^{(t)}(x_{sb_0+[s]})$, degree-$1$ entries from $\cK_{b_1}^{(t)}(x_{sb_1+[s]})$, and so on.
	Therefore, we can conclude that $\det(\cK_{B})$ is a homogeneous polynomial of degree $(0+1+\dots+t-1)s=\frac{st(t-1)}{2}.$
	In addition, the largest power of $\lambda_i$ occurring in any monomial of $\det(\cK_{B})$ is $t-1$.
	So that $D$ is a homogeneous polynomial 
	of degree $$\deg(D) = \sum_{t=1}^{s}{\binom{s}{t}}\frac{st(t-1)}{2} = s^2(s-1)2^{s-3},$$ and the largest power of $\lambda_i, 1\leq i\leq s^2-1$, occurring in any monomial of $D$ is $$\deg_{x_i}(D) = \sum_{t=1}^{s}{\tbinom{s-1}{t-1}}(t-1) = (s-1) 2^{s-2}.$$
	Applying Corollary~\ref{corollary:2} to $D(x_{[s^2]})$,
	we can find $s^2$ distinct elements $\lambda_0, \lambda_1, \dots, \lambda_{s^2-1}$ in any subset $S\subseteq\Fq$ with $|S|\ge s^2 + (s-1) 2^{s-2}$ satisfying
	$D(\lambda_{[s^2]})\neq0$.
	
	Suppose $q \ge ns + (s-1) 2^{s-2}$.
	We first choose $s^2$ distinct elements $\lambda_{[s^2]}$ from $\Fq$ satisfying
	the local constraints for group $0$.
	Then we choose $\lambda_{s^2+[s^2]}$ from $\Fq\setminus \lambda_{[s^2]}$ satisfying
	the local constraints for group $1$.
	As $q \ge ns + (s-1) 2^{s-2}$, we can repeat this process $n/s$ times.
	
	Now we prove the latter part of Lemma~\ref{lemma:FieldSizeA}.
	Consider the polynomial $$\Phi(x) = D(1,x,x^2\dots, x^{s^2-1}),$$ i.e, we assign the value $x_i = x^i$ in the polynomial $D(x_{[s^2]})$.
	Firstly, we claim that $\Phi(x)$ is a nonzero polynomial in $\Fq[x]$.
	We know that
	\begin{align*}
		\Phi(x) = & D(1,x,\dots, x^{s^2-1})                                \\
		=         & \prod_{B\subseteq [s]}\cK_B\big(x^i, i\in Bs+[s]\big).
	\end{align*}
	If $B = \{b_0, \dots, b_{t-1}\}$ and $b_0<\dots<b_{t-1}$,
	then using the same method as \eqref{eq:KBterm}, we can verify that the term
	$$
	\prod_{z\in [t]}\prod_{j\in [s]}(x^{b_zs+j})^z  = x^{d_B},
	$$
	where $d_{B} = \sum_{z\in [t], j\in [s]}z( b_zs+j)$,
	is the unique highest term in the polynomial $\cK_B(x^i, i\in Bs+[s])$.
	Thus, $\Phi(x)$ is a nonzero polynomial.
	
	Suppose that $q\geq ns+1$ and $n$ is large enough compared with $s$ such that $$\varphi(q-1) > \deg(\Phi(x))\ge 0$$ where $\varphi(x)$ is the Euler function.
	Since $\deg(\Phi(x))$ only depends on the value of $s$,  the above condition can hold if $n$ is large enough compared with $s$.
	Note that there are $\varphi(q-1)$ primitive elements in $\Fq$, so we can find some primitive element $\alpha$ of $\Fq$ such that $$\Phi(\alpha) = D(1, \alpha , \alpha^2, \dots, \alpha^{s^2-1})\neq 0.$$
	Since $D(x_{[s^2]})$ is a homogeneous polynomial, we know that
	\begin{align*}
		& D(\alpha^{as^2} \cdot (1,\alpha, \alpha^2, \dots, \alpha^{s^2-1}))               \\
		= & \alpha^{as^2\deg(D)} \cdot  D(1,\alpha, \alpha^2, \dots, \alpha^{s^2-1}) \neq 0.
	\end{align*}
	In other words, if we set $\lambda_i = \alpha^i$ for all $i\in [ns]$.
	Then these $\lambda_i$s satisfy the local constraints~\eqref{eq:localA}.

	This concludes our proof of Lemma~\ref{lemma:FieldSizeA}.

	
	\bibliographystyle{ieeetr}
	\bibliography{msr}

\begin{thebibliography}{10}

\bibitem{Dimakis10}
A.~G. Dimakis, P.~B. Godfrey, Y.~Wu, M.~J. Wainwright, and K.~Ramchandran,
  ``Network coding for distributed storage systems,'' {\em IEEE Transactions on
  Information Theory}, vol.~56, no.~9, pp.~4539--4551, 2010.

\bibitem{Guruswami16STOC}
V.~Guruswami and M.~Wootters, ``Repairing {Reed-Solomon} codes,'' in {\em
  Proceedings of the Forty-Eighth Annual ACM Symposium on Theory of Computing},
  STOC '16, (New York, NY, USA), p.~216–226, Association for Computing
  Machinery, 2016.

\bibitem{Guruswami16}
V.~Guruswami and M.~Wootters, ``Repairing {R}eed-{S}olomon codes,'' {\em IEEE
  Transactions on Information Theory}, vol.~63, no.~9, pp.~5684--5698, 2017.

\bibitem{Tamo17RS}
I.~{Tamo}, M.~{Ye}, and A.~{Barg}, ``Optimal repair of {R}eed-{S}olomon codes:
  {A}chieving the cut-set bound,'' in {\em 2017 IEEE 58th Annual Symposium on
  Foundations of Computer Science (FOCS)}, pp.~216--227, Oct 2017.

\bibitem{Rashmi11}
K.~V. Rashmi, N.~B. Shah, and P.~V. Kumar, ``Optimal exact-regenerating codes
  for distributed storage at the {MSR} and {MBR} points via a product-matrix
  construction,'' {\em IEEE Transactions on Information Theory}, vol.~57,
  no.~8, pp.~5227--5239, 2011.

\bibitem{Tamo13}
I.~Tamo, Z.~Wang, and J.~Bruck, ``Zigzag codes: {MDS} array codes with optimal
  rebuilding,'' {\em IEEE Transactions on Information Theory}, vol.~59, no.~3,
  pp.~1597--1616, 2013.

\bibitem{Wang16}
Z.~Wang, I.~Tamo, and J.~Bruck, ``Explicit minimum storage regenerating
  codes,'' {\em IEEE Transactions on Information Theory}, vol.~62, no.~8,
  pp.~4466--4480, 2016.

\bibitem{Ye16}
M.~Ye and A.~Barg, ``Explicit constructions of high-rate {MDS} array codes with
  optimal repair bandwidth,'' {\em IEEE Transactions on Information Theory},
  vol.~63, no.~4, pp.~2001--2014, 2017.

\bibitem{Ye16a}
M.~Ye and A.~Barg, ``Explicit constructions of optimal-access {MDS} codes with
  nearly optimal sub-packetization,'' {\em IEEE Transactions on Information
  Theory}, vol.~63, no.~10, pp.~6307--6317, 2017.

\bibitem{Sasid16}
B.~Sasidharan, M.~Vajha, and P.~V. Kumar, ``An explicit, coupled-layer
  construction of a high-rate {MSR} code with low sub-packetization level,
  small field size and all-node repair,'' {\em arXiv preprint
  arXiv:1607.07335}, 2016.

\bibitem{Raviv17}
N.~Raviv, N.~Silberstein, and T.~Etzion, ``Constructions of high-rate minimum
  storage regenerating codes over small fields,'' {\em IEEE Transactions on
  Information Theory}, vol.~63, no.~4, pp.~2015--2038, 2017.

\bibitem{Li18}
J.~Li, X.~Tang, and C.~Tian, ``A generic transformation to enable optimal
  repair in {MDS} codes for distributed storage systems,'' {\em IEEE
  Transactions on Information Theory}, vol.~64, no.~9, pp.~6257--6267, 2018.

\bibitem{Duursma21}
I.~Duursma and H.~Wang, ``Multilinear algebra for minimum storage regenerating
  codes: a generalization of the product-matrix construction,'' {\em Applicable
  Algebra in Engineering, Communication and Computing}, pp.~1--27, 2021.

\bibitem{Vajha21}
M.~Vajha, S.~B. Balaji, and P.~Vijay~Kumar, ``Small-d {MSR} codes with optimal
  access, optimal sub-packetization, and linear field size,'' {\em IEEE
  Transactions on Information Theory}, vol.~69, no.~7, pp.~4303--4332, 2023.

\bibitem{CIT-115}
V.~Ramkumar, S.~B. Balaji, B.~Sasidharan, M.~Vajha, M.~N. Krishnan, and P.~V.
  Kumar, ``Codes for distributed storage,'' {\em Foundations and Trends® in
  Communications and Information Theory}, vol.~19, no.~4, pp.~547--813, 2022.

\bibitem{Goparaju14}
S.~Goparaju, I.~Tamo, and R.~Calderbank, ``An improved sub-packetization bound
  for minimum storage regenerating codes,'' {\em IEEE Transactions on
  Information Theory}, vol.~60, no.~5, pp.~2770--2779, 2014.

\bibitem{Alrabiah19}
O.~Alrabiah and V.~Guruswami, ``An exponential lower bound on the
  sub-packetization of {MSR} codes,'' in {\em Proceedings of the 51st Annual
  ACM SIGACT Symposium on Theory of Computing}, STOC 2019, (New York, NY, USA),
  p.~979–985, Association for Computing Machinery, 2019.

\bibitem{Tamo14}
I.~Tamo, Z.~Wang, and J.~Bruck, ``Access versus bandwidth in codes for
  storage,'' {\em IEEE Transactions on Information Theory}, vol.~60, no.~4,
  pp.~2028--2037, 2014.

\bibitem{Balaji22}
S.~B. Balaji, M.~Vajha, and P.~Vijay~Kumar, ``Lower bounds on the
  sub-packetization level of {MSR} codes and characterizing optimal-access
  {MSR} codes achieving the bound,'' {\em IEEE Transactions on Information
  Theory}, vol.~68, no.~10, pp.~6452--6471, 2022.

\bibitem{Rawat16}
A.~S. Rawat, O.~O. Koyluoglu, and S.~Vishwanath, ``Progress on high-rate {MSR}
  codes: Enabling arbitrary number of helper nodes,'' in {\em 2016 Information
  Theory and Applications Workshop (ITA)}, pp.~1--6, 2016.

\bibitem{LWHY22}
N.~Wang, G.~Li, S.~Hu, and M.~Ye, ``Constructing {MSR} codes with
  subpacketization $2^{n/3}$ for $k + 1$ helper nodes,'' {\em IEEE Transactions
  on Information Theory}, vol.~69, no.~6, pp.~3775--3792, 2023.

\bibitem{Liu22}
Y.~Liu, J.~Li, and X.~Tang, ``A generic transformation to generate {MDS} array
  codes with $\delta$-optimal access property,'' {\em IEEE Transactions on
  Communications}, vol.~70, no.~2, pp.~759--768, 2022.

\bibitem{Zhang23}
Z.~Zhang and L.~Zhou, ``A vertical-horizontal framework for building rack-aware
  regenerating codes,'' {\em IEEE Transactions on Information Theory}, vol.~69,
  no.~5, pp.~2874--2885, 2023.

\bibitem{Li23}
J.~Li, Y.~Liu, X.~Tang, Y.~S. Han, B.~Bai, and G.~Zhang, ``{MDS} array codes
  with (near) optimal repair bandwidth for all admissible repair degrees,''
  {\em IEEE Transactions on Communications}, vol.~71, no.~10, pp.~5633--5646,
  2023.

\bibitem{CB}
Z.~Chen and A.~Barg, ``Explicit constructions of {MSR} codes for clustered
  distributed storage: The rack-aware storage model,'' {\em IEEE Transactions
  on Information Theory}, vol.~66, no.~2, pp.~886--899, 2020.

\bibitem{tensorpermutation}
H.~V. Henderson and S.~R. Searle, ``The vec-permutation matrix, the vec
  operator and kronecker products: A review,'' {\em Linear and multilinear
  algebra}, vol.~9, no.~4, pp.~271--288, 1981.

\bibitem{alon_1999}
N.~Alon, ``Combinatorial nullstellensatz,'' {\em Combinatorics, Probability and
  Computing}, vol.~8, no.~1-2, p.~7–29, 1999.

\end{thebibliography}

\end{document}